	\documentclass[journal]{IEEEtran}
	\newcommand{\figWidth}{0.8}

\usepackage[utf8]{inputenc}
\usepackage{amsfonts}
\usepackage{amssymb}
\usepackage{subfigure}
\usepackage{amsmath}
\usepackage{amsthm}
\usepackage{mathrsfs}
\usepackage{verbatim}
\usepackage{graphicx}
\usepackage{epstopdf}
\usepackage{enumitem}
\usepackage{caption}
\usepackage{xcolor}

\theoremstyle{remark} \newtheorem{remark}{Remark}

\newtheorem{proposition}{Proposition}
\newtheorem{definition}{Definition}

\newtheorem{corollary}{Corollary}

\newcommand{\thmlabel}[1]{\label{thm:#1}}
\newcommand{\thmref}[1]{\ref{thm:#1}}
\newcommand{\Thmref}[1]{Thm.~\thmref{#1}}

\theoremstyle{remark} \newtheorem{theorem}{Theorem}
\newtheorem{example}{Example}

\newcommand{\xvec}{\mathbf{x}}
\newcommand{\yvec}{\mathbf{y}}

\newcommand{\Yvec}{\mathbf{Y}}

\newcommand{\mX}{\mathcal{X}}

\newcommand{\Lev}{\mathscr{L}}
\newcommand{\IG}{\mathscr{IG}}

\newcommand{\realSet}{\mathcal{R}}

\newcommand{\E}{\mathbb{E}}

\newcommand{\erfc}{\mathsf{erfc}}

\newcommand{\sgn}{\mathsf{sgn}}

\newcommand{\PhiG}{\Phi_{\text{G}}}

\newcommand{\argmax}{\operatornamewithlimits{argmax}}

\newcommand{\tyony}[1]{{\leavevmode\color{black}#1}}
\newcommand{\changeYony}[1]{{\leavevmode\color{black}#1}}

\begin{document}

\title{Optimal Detection for Diffusion-Based \\ Molecular Timing Channels}

\author{ \IEEEauthorblockN{Yonathan~Murin,~\IEEEmembership{Member,~IEEE,}
		Nariman~Farsad,~\IEEEmembership{Member,~IEEE,}\\
		Mainak Chowdhury,~\IEEEmembership{Student Member,~IEEE,} 
		and Andrea J. Goldsmith,~\IEEEmembership{Fellow,~IEEE}}

\thanks{The authors are with the Department of Electrical Engineering, Stanford University, Stanford, CA, 94305 USA. This work was presented in part at the IEEE Global Communication Conference (GLOBECOME), December 2016, Washington DC, USA, \cite{murinGlobeCom16}. This research was supported in part by the NSF Center for Science of Information (CSoI) under grant CCF-0939370, and the NSERC Postdoctoral Fellowship fund PDF-471342-2015.}		
}
\date{}


\maketitle
\thispagestyle{empty} 

\vspace{-0.65cm}
\begin{abstract}

	This work studies optimal detection for communication over diffusion-based molecular timing (DBMT) channels. 
	The transmitter {\em simultaneously} releases {\em multiple} information particles, where the information is encoded in the time of release. The receiver decodes the transmitted information based on the random time of arrival of the information particles, which is modeled as an additive noise channel. For a DBMT channel without flow, this noise follows the L\'evy distribution. Under this channel model, the maximum-likelihood (ML) detector is derived and shown to have high computational complexity. It is also shown that under ML detection, releasing multiple particles improves performance, while for any additive channel with $\alpha$-stable noise where $\alpha<1$ (such as the DBMT channel), under linear processing at the receiver, releasing multiple particles {\em degrades} performance {\em relative to releasing a single particle}. 
	Hence, a new low-complexity detector, which is based on the first arrival (FA) among all the transmitted particles, is proposed. It is shown that for a small number of released particles, the performance of the FA detector is very close to that of the ML detector. On the other hand, error exponent analysis shows that the performance of the two detectors differ when the number of released particles is large.
	
\end{abstract}

\section{Introduction}	

\IEEEPARstart{M}{olecular} communication (MC) is an emerging field in which nano-scale devices communicate with each other via chemical signaling, based on exchanging small {\em information particles} \cite{eckBook, far16ST}. 
For instance, in biological systems MC can take place using hormones, pheromones, or ribonucleic acid molecules. 
To embed information in these particles one may use the particle's type \cite{kim13}, concentration \cite{kur12, nak12}, number \cite{far15TNANO}, or the time of release \cite{eck07, ITsubmission}. Particles can be transported from the transmitter to the receiver via diffusion, active transport, bacteria, and flow, as described in \cite[Sec. III.B]{far16ST} and the references therein. Although this new field is still in its infancy, several basic experimental systems serve as a proof of concept for transmitting short messages at low bit rates \cite{far13,lee2015_infocom,koo16}. 

There are several similarities between traditional electromagnetic (EM) communication and MC. As a result, several prior works have used tools and algorithms developed for EM communication in the design of MC systems. 
In particular, the work \cite{kil13} studied on-off transmission via diffusion of information particles, where the information is recovered at the receiver based on the measured concentration. A channel model with finite memory was proposed, which involves additive Gaussian noise, along with several sequence detection algorithms such as maximum a-posteriori (MAP) detection and maximum likelihood (ML) detection. The work \cite{Meng14} studied a similar setup proposing a technique for inter-symbol interference (ISI) mitigation and deriving a reduced-state ML sequence detection algorithm. 
Finally, \cite{noe14_RxDesign} studied on-off transmission over diffusive molecular channels with flow, proposed an ML sequence detection algorithm for this channel, and \changeYony{designed a family of sub-optimal weighted sums detectors with relatively low complexity}. 
While the above works build upon the similarities between EM communication and MC, namely, linear channel models with additive (and in some cases Gaussian) noise, there are aspects in which MC is fundamentally different from traditional EM communication. 
For instance, in EM communication the symbol duration is fixed, while in MC the symbol duration is often a random variable (RV).
Therefore, information particles may arrive out-of-order, which makes correctly detecting particles in the order in which they were transmitted very challenging, in particular when the transmitted information particles are indistinguishable \cite{rose15, rose:InscribedPart1, nanoComNet}.

This work focuses on receiver design for MC systems where information is modulated through the {\em time of release of the information particles}, which is reminiscent of pulse position-modulation \cite{shiu99}. A common assumption, which is accurate for many sensors, is that after some time duration each particle is absorbed by the receiver and removed from the environment. In this case, the random delay until a released particle arrives at the receiver can be modeled as a channel with an additive noise term. For diffusion-based channels {\em without flow}, this additive noise is L\'evy-distributed \cite{yilmaz20143dChannelCF}, while for diffusion-based channels {\em with flow}, this additive noise follows an inverse Gaussian (IG) distribution \cite{sri12}. Fig.~\ref{fig:diffuseMolComm} illustrates the additive noise timing channel model studied in this work. 

At first glance, the cases of diffusion with and without flow may seem similar; however, a closer look reveals a fundamental difference which stems from the different properties of the additive noise modeling the random propagation delay of each particle. The L\'evy distribution has an algebraic tail\footnote{An RV $X$ has an algebraic tail if there exists $\rho_1,\rho_2 > 0$ such that $\lim_{x \to \infty} x^{\rho_2} \Pr \{ |X| > x \} = \rho_1$.} \cite{nol15, gonzales07}, while the tail of the IG distribution, similarly to the standard Gaussian distribution, decays exponentially. Thus, traditional linear detection and signal processing techniques, which work well in the presence of noise with exponentially-decaying  distributions such as Gaussian or IG noise, may perform poorly in the presence of additive L\'evy noise. 
The need for new detection methods in communication systems operating over channels with additive noise, characterized by algebraic tails, was observed in \cite{nikias-book} based on numerical simulations.
\changeYony{In this work we rigorously prove that, when multiple particles are simultaneously released, the detection performance in diffusion-based molecular timing (DBMT) channels {\em without} flow {\em cannot be improved} by linear processing, compared to optimal detection when a single particle is released. While in the case of the DBMT channel without flow the noise is L\'evy distributed, thus belonging to the family of $\alpha$-stable distributions \cite{nol15, zolotarev-book}, our result regarding the inefficiency of linear processing extends to any $\alpha$-stable noise with $\alpha < 1$.}
We note that $\alpha$-stable distributions are commonly used to model impulsive noise \cite{nikias-book, yang14, Fah_arXiv16}. Yet, the focus of the studies \cite{nikias-book, yang14, Fah_arXiv16} was on {\em symmetric} stable distributions. On the other hand, in this work we focus on the DBMT channel without flow, in which the additive noise follows the {\em asymmetric} L\'evy distribution.

In addition to the fact that the tails of the additive noise decay slowly, ordering in time is not preserved in the considered diffusion-based timing channel. In particular, the information particles associated with a given symbol may arrive later than particles associated with a subsequent symbol. This gives rise to ISI. \changeYony{In the works \cite{nanoComNet, nanocom16} we designed a sequence detector for time-slotted transmission over DBMT channels without flow, when a single information particle is used per symbol.
In this work, on the other hand, we focus on systems for which the ISI is negligible and multiple information particles are used to modulate a symbol. This setting arises in molecular communication systems with long symbol times such that the propagation delay of information particles is typically less than a symbol time.}
Negligible ISI also arises in systems with one-shot bursty communication, such as a sensor that occasionally sends a single symbol conveying one or more bits, and then remains silent for many symbol times.
Since we neglect ISI in our model, each symbol transmission can be analyzed independently.

Specifically, we consider an MC system in which the information is encoded in the time of release of the information particles, where this time is selected out of a set with {\em finite cardinality}, namely, a {\em finite constellation} is used. 
At each transmission {\em $M$ information particles are simultaneously} released at the time corresponding to the current symbol, while the receiver's objective is to detect this transmission time. Note that $M$ is constant and does not change from one transmission to the next, i.e., information is not encoded in the number of particles. The $M$ particles travel over a DBMT channel without flow. We assume that consecutive channel uses are independent and identically distributed (i.i.d.). 

We derive the ML detection rule for our system which, as expected, entails high computational complexity. 
This motivates studying detectors with lower complexity. 
A common approach to reducing detector complexity in traditional EM communication, which was also proposed in \cite{sri12} for an MC system, is to use a linear detector. 
\changeYony{We show that for $M$ i.i.d samples of {\em any} $\alpha-$stable additive noise with $\alpha<1$, and in particular for L\'evy-distributed noise, linearly combining these samples results in a $\alpha-$stable RV with dispersion larger or equal to the dispersion of the original samples. Here dispersion is a parameter of the distribution measuring its spread,\footnote{The variance of a stable distribution RV with $\alpha<2$ is infinite.} see \cite[Defs. 1.7 and 1.8]{nol15}.}
This increased dispersion degrades the probability of correct detection, compared to the case of a {\em single particle}. In other words, a linear detector in our system has better performance when a single particle is used to convey the symbol time ($M=1$) compared to when multiple particles convey the symbol time ($M>1$).
To the best of our knowledge this is the first proof that linear processing degrades the performance of multiple particle release relative to single particle release in a MC system. 

In order to take advantage of multiple transmitted particles per symbol, we propose a new detector based on the first arrival (FA) among the $M$ information particles. 
\changeYony{We show that the probability density of the FA, conditioned on the transmitted time, concentrates towards the transmission time when $M$ increases, see Fig. \ref{fig:CondDistributions} in Section \ref{subsec:FA}. This increases the probability of correct detection, compared to the case of a single particle. This is in contrast to the probability density of a linear combination of the arrival times, conditioned on the transmission time, which disperses from the transmission time compared to the case of a single particle.}
Furthermore, we show that the performance of the proposed FA detector is very close to that of the optimal ML detector, for small values of $M$ \changeYony{(on the order of tens)}.
On the other hand, we use error exponent analysis to show that for large values of $M$, i.e., $M \to \infty$, ML significantly outperforms the FA detector, which agrees with the fact that the FA is {\em not} a sufficient statistic for the arrival time of all $M$ transmitted particles, as is computed by the ML detector.  

The rest of this paper is organized as follows.
The problem formulation is presented in Section \ref{sec:ProbForm}. 
Sections \ref{sec:DBMT}--\ref{subsec:PerfCompare} study the case of a binary constellation \changeYony{($M$ particles are simultaneously released in one of two pre-defined timings)}: The ML detector and linear detection are studied in Section \ref{sec:DBMT}. The FA detector is derived in Section \ref{subsec:FA}, while its performance is compared to the performance of the ML detector in Section \ref{subsec:PerfCompare}. The FA detector is extended to the case of larger constellations in Section \ref{sec:beyondBinary}. 
Numerical results are presented in Section \ref{sec:numRes}, and concluding remarks are provided in Section \ref{sec:conc}.

{\bf {\slshape Notation}:} We denote the set of real numbers by $\realSet$, the set of positive real numbers by $\realSet^{+}$, and the set of integers by $\mathcal{N}$. Other than these sets, we denote sets with calligraphic letters, e.g., $\mathcal{B}$.
We denote RVs with upper case letters, e.g., $X$, $Y$, and their realizations with lower case letters, e.g., $x$, $y$. 
An RV takes values in the set $\mX$, and we use $|\mX|$ to denote the cardinality of a finite set. 
We use $f_{Y}(y)$ to denote the probability density function (PDF) of a continuous RV $Y$ on $\realSet$, $f_{Y|X}(y|x)$ to denote the conditional PDF of $Y$ given $X$, and $F_{Y|X}(y|x)$ to denote the conditional cumulative distribution function (CDF). 
We denote vectors with boldface letters, e.g., $\xvec, \yvec$, where the $k^{\text{th}}$ element of a vector $\xvec$ is denoted by $x_k$.
Finally, we use $\PhiG(x) = \frac{1}{\sqrt{2\pi}} \int_{-\infty}^{x}{e^{-u^2} du}$ to denote the CDF of a standard Gaussian RV, $\erfc(x) = \frac{2}{\sqrt{\pi}} \int_{x}^{\infty}{e^{-u^2} du}$ to denote the complementary error function, $\log (\cdot)$ to denote the natural logarithm, and $\E\{\cdot\}$ to denote stochastic expectation.

\section{Problem Formulation} \label{sec:ProbForm}

\subsection{System Model} \label{subsec:sysModel}

Fig. \ref{fig:diffuseMolComm} illustrates a molecular communication channel in which information is modulated on {\em the time of release of the information particles}.  
We assume that the information particles themselves are {\em identical and indistinguishable} at the receiver. Therefore, the receiver can only use the time of arrival to decode the intended message.
The information particles propagate from the transmitter to the receiver through some random propagation mechanism (e.g. diffusion). We make the following assumptions about the system:

	\renewcommand{\figWidth}{0.8}

\begin{figure}[t]
	\begin{center}
		\includegraphics[width=\figWidth\columnwidth,keepaspectratio]{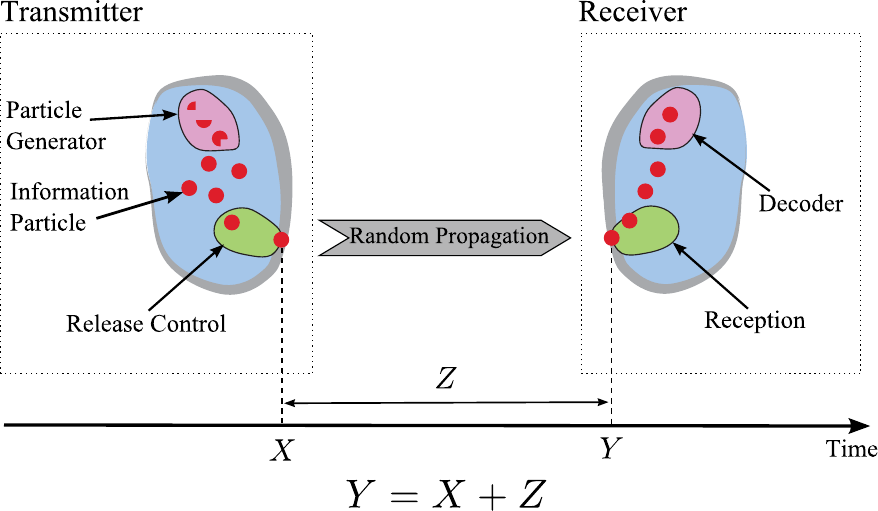}
	\end{center}
	\vspace{-0.3cm}
	\captionsetup{font=footnotesize}
	\caption{\label{fig:diffuseMolComm} Diffusion-based molecular communication timing channel. $X$ denotes the release time, $Z$ denotes the random propagation time, and $Y$ denotes the arrival time.}
	\vspace{-0.25cm}
\end{figure}

\begin{enumerate}[label = {\bf A{\arabic*}})]
	\item \label{assmp:synch}
	The transmitter perfectly controls the release time of each information particle, and the receiver perfectly measures the arrival times of the information particles. 
	Moreover, the transmitter and the receiver are perfectly synchronized in time.		
	
	\item \label{assmp:Arrival}
	An information particle that arrives at the receiver is absorbed and removed from the propagation medium.
	
	\item \label{assmp:indep}
	All information particles propagate independently of each other, and their trajectories are random according to an i.i.d. random process. This is a reasonable assumption for many different propagation schemes in molecular communication such as diffusion in dilute solutions, i.e., when the number of particles released is much smaller than the number of molecules of the solutions.
\end{enumerate} 

\noindent Note that these assumptions have been traditionally considered in all previous works, e.g. \cite{nak12,ITsubmission,ata13,pie13, li14}, in order to make the models tractable.

%
		%
		%
		%

Let $\mX$ be a finite set of constellation points on the real line: $\mathcal{X} \triangleq \{\xi_0, \xi_1, \dots, \xi_{L-1} \}$, $0 \le \xi_0 \le \dots \le \xi_{L-1}$, and let $\xi_{L-1} < T_s < \infty$ denote the symbol duration. The $k^{\text{th}}$ transmission takes place at time $(K-1)T_s + X_k, X_k \in \mX, k=1,2,\dots,K$.
At this time, $M \in \mathcal{N}$ information particles are {\em simultaneously} released into the medium by the transmitter. 
We assume that at each transmission the same number of information particles is released.
The transmitted information is encoded in the sequence $\{(K-1)T_s + X_k \}_{k=1}^K$, which is assumed to be independent of the random propagation time of {\em each} of the information particles. 
Let $\Yvec_k$ denote an $M$-length vector consisting of the times of arrival of each of the information particles released at time $(k-1)T_s + X_k$. It follows that $Y_{k,m} > X_k, m=1,2,\dots,M$. Thus, we obtain the following additive noise channel model:     
\begin{align}
	Y_{k,m} = (k-1)T_s + X_k + Z_{k,m},
\label{eq:LevyChan_k}
\end{align}

\noindent for $k \mspace{-2mu} = \mspace{-2mu} 1,2,\dots,K, m \mspace{-2mu} = \mspace{-2mu}1,2,\dots,M$, where $Z_{k,m}$ is a random noise term representing the propagation time of the $m^{\text{th}}$ particle of the $k^{\text{th}}$ transmission. Note that Assumption \ref{assmp:indep} implies that all the RVs $Z_{k,m}$ are independent. 

In the channel model \eqref{eq:LevyChan_k}, particles may arrive out of order, which results in a channel with memory. In this work, however, we assume that each information particle arrives before the next transmission takes place. This assumption can be formally stated as:
\begin{enumerate}[label = {\bf A{\arabic*}}), resume]

	\item \label{assmp:memoryless}
	$T_s$ is a fixed constant chosen to be large enough such that the transmission times $X_k$ obey $Y_{k,m} \le k T_s$ with high probability.\footnote{Formally, let $\eta$ be arbitrarily high probability, then we choose $T_s$ such that $\Pr \{ Y_{k,m} < k T_s \} > \eta, k=1,2,\dots,K, m=1,2,\dots,M$.}
	
\end{enumerate}

\noindent With this assumption, we obtain an i.i.d. memoryless channel model which can be written as:
\begin{align}
	Y_{m} = {X} + Z_{m}, \quad m=1,2,\dots,M.
\label{eq:LevyChan}
\end{align} 

\noindent In the rest of this work we focus on this memoryless channel model.

Assumption \ref{assmp:memoryless} implies that $T_s$ is chosen such that consecutive transmissions are sufficiently separated in time relative to the random propagation delays of each particle.
Thus, the effective communication channel is memoryless.
 
To simplify the presentation, in most of this work we restrict our attention to the case of binary modulation, i.e., $\mX = \{ \xi_0, \xi_1 \}$. 
Higher order modulations are discussed in Section \ref{sec:beyondBinary}.
Let $S \in \{ 0,1 \}$, be an equiprobable bit to be sent over the channel \eqref{eq:LevyChan} to the receiver, and denote the estimate of $S$ at the receiver by $\hat{S}$. We note that our results can be easily extended to the case of different a-priori probabilities on the transmitted bits. 
Our objective is to design a receiver that minimizes the probability of error $P_{\varepsilon} = \Pr \{S \neq \hat{S} \}$.  
In order to minimize $P_{\varepsilon}$ we maximize the spacing between $\xi_0$ and $\xi_1$, and without loss of generality we use the following mapping for transmission:
\begin{equation}
	X(S) = \begin{cases} 0, & s = 0 \\ \Delta , & s = 1. \end{cases}
\label{eq:TxMapping}
\end{equation}

\noindent Note that the above description of communication over an MT channel is fairly general and can be applied to different propagation mechanisms as long as Assumptions \ref{assmp:synch}--\ref{assmp:memoryless} are not violated. Next, we describe the DBMT channel.

\subsection{The DBMT Channel} \label{subsec:DBMTdef}

In diffusion-based propagation, the released particles follow a random Brownian path from the transmitter to the receiver. In this case, to specify the random additive noise term $Z_m$ in \eqref{eq:LevyChan}, we define a L\'evy-distributed RV as follows:
\begin{definition}[{\em L\'evy distribution}]
Let $Z$ be L\'evy-distributed with location parameter $\mu$ and scale parameter $c$ \cite{nol15}. Then its PDF is given by:
	\begin{align}
	\label{eq:LevyNoise}
	f_Z(z)=
	\begin{cases}
	\sqrt{\frac{c}{2 \pi (z-\mu)^3}}\exp \left( -\frac{c}{2(z-\mu)} \right), & z>\mu \\
	0, & z\leq \mu
	\end{cases},
	\end{align} 
	
	\noindent and its CDF is given by:
	\begin{align}
	\label{eqn:LevyCDF}
	F_Z(z) = \begin{cases} \erfc\left(\sqrt{\frac{c}{2(z-\mu)}}\right), & z>\mu \\ 0, & z\leq\mu \end{cases}.		
	\end{align}

\end{definition}
	
	Let $d$ denote the distance between the transmitter and the receiver, and $D$ denote the diffusion coefficient of the information particles in the propagation medium. Following along the lines of the derivations in \cite[Sec. II]{sri12}, and using the results of \cite[Sec. 2.6.A]{karatzas-shreve}, it can be shown that for 1-dimensional pure diffusion, the propagation time of each of the information particles follows a L\'evy distribution, denoted in this work by $\sim \Lev(\mu,c)$ with $c = \frac{d^2}{2D}$ and $\mu=0$. Thus, $Z_m \sim \Lev(0,c), m=1,2,\dots,M$.
	\changeYony{Note that the scale parameter $c$ increases quadratically with the distance between the transmitter and the receiver $d$, and behaves inversely linear with the diffusion coefficient $D$, that has units of squared meter per second. Thus, the scale parameter $c$ has units of seconds.}

\begin{remark}[{\em Scaled L\'evy distribution for 3-D space}]
	The work \cite{yilmaz20143dChannelCF} showed that a scaled L\'evy distribution can also model the first arrival time in the case of an infinite, three-dimensional homogeneous medium without flow. Hence, our results can be extended to 3-D space by simply introducing a scalar factor. 
\end{remark}



The L\'evy distribution belongs to the class of stable distributions, discussed in the next subsection. For a detail description we refer the reader to \cite{nol15, zol86-book}. 

\subsection{Stable Distributions}

\begin{definition}[{\em Stable distribution}]
	An RV $X$ has a stable distribution if for two independent copies of $X$, $X_1$ and $X_2$, and for any constants $a_1, a_2 \in \realSet^{+}$, there exists constants $a_3 \in \realSet^{+}$ and $a_4 \in \realSet$ such that:
	\begin{align}
		a_1 X_1 + a_2 X_2 \stackrel{d}{=} a_3 X + a_4, \label{eq:stableDistDef}
	\end{align}
	
	\noindent where $\stackrel{d}{=}$ denotes equality in distribution, i.e., both expressions follow the same probability law. 
\end{definition}

Stable distributions can also be defined via their characteristic function.
\begin{definition}[{\em Characteristic function of a stable distribution}]
	Let $-\infty < \mu < \infty, c\ge 0, 0 < \alpha \le 2$, and $-1 \le \beta \le 1$. Further define: 
	\begin{align*}
		\Phi(t,\alpha) \triangleq \begin{cases} 
												\tan \left( \frac{\pi \alpha}{2}\right), & \alpha \ne 1 \\ 
												-\frac{2}{\pi} \log (|t|), & \alpha = 1
											\end{cases}.
	\end{align*}
	
	\noindent Then, the characteristic function of a stable RV $X$, with location parameter $\mu$, scale (or dispersion) parameter $c$, characteristic exponent $\alpha$, and skewness parameter $\beta$, is given by:
	\begin{align}
		\varphi(t;\mu,c,\alpha,\beta) \mspace{-3mu} = \mspace{-3mu} \exp \left\{ j \mu t \mspace{-3mu} - \mspace{-3mu} |ct|^\alpha (1 \mspace{-3mu} - \mspace{-3mu} j \beta \sgn(t) \Phi(t,\alpha))\right\}. \label{eq:stableCharFunc}
	\end{align}
\end{definition}
	
In the following, we use the notation $\mathscr{S}(\mu, c, \alpha, \beta)$ to represent a stable distribution with the parameters $\mu, c, \alpha$, and $\beta$.
Apart from several special cases, stable distributions do not have closed-form PDFs. The exceptional cases are the Gaussian distribution ($\alpha = 2$), the Cauchy distribution $(\alpha = 1)$, and the case of $\alpha = \frac{1}{2}$ which was very recently derived in \cite[Theorem 2]{NGCE:15}. Note that the L\'evy distribution is a special case of the results of \cite{NGCE:15} with $\beta = 1$, i.e,  the L\'evy distribution belongs to the class of stable distributions with the parameters $\mathscr{S}(\mu, c, \frac{1}{2}, 1)$. Thus, its characteristic function is given by:
	\begin{align*}
		\varphi(t) = \exp \left\{ j \mu t - \sqrt{-2jct} \right\}.
	\end{align*}

\noindent Finally, we note that all stable distributions, apart from the case $\alpha = 2$, have infinite variance, and all stable distributions with $\alpha \le 1$ also have infinite mean. In fact, this statement can be generalized to moments of order $p \le \alpha$, see \cite{gonzales07}.

Next, we study ML and linear detection of particle arrival time for transmission over the DBMT channel.

\section{Transmission over the DBMT Channel: ML and Linear Detection} \label{sec:DBMT}

\subsection{Transmission over the Single-Particle DBMT Channel} \label{sec:SingleMolecule}

We begin this section with the relatively simple case in which a single information particle is released, i.e., $M=1$. For this setup, the decision rule that minimizes the probability of error, and the corresponding minimal probability of error, are given in the following proposition:
\begin{proposition}
\label{prop:decRuleSymbBySymb}
	The decision rule that minimizes the probability of error when $M=1$, is given by:
	\begin{align}
		\hat{S}_{\text{ML}}(y_1) = \begin{cases} 0, & y_1 < \theta \\ 1, &  y_1 \ge \theta, \end{cases}
		\label{eq:decisionRule}
	\end{align}
	
	\noindent where $\theta$ is the unique solution, in the interval $[\Delta, \Delta + \frac{c}{3}]$, of the following equation in $y_1$:
	\begin{align}
		y_1(y_1-\Delta) \log \left( \frac{y_1}{y_1-\Delta} \right) = \frac{c \Delta}{3}, \quad y_1 > \Delta > 0.
		\label{eq:thetaEquation}
	\end{align}
	
	\noindent Furthermore, the probability of error of this decision rule is given by:
	\begin{align}
		P_{\varepsilon} = 0.5 \left(1 - \erfc \left( \sqrt{\frac{c}{2\theta}} \right) + \erfc \left( \sqrt{\frac{c}{2(\theta - \Delta)}} \right) \right).
		\label{eq:errProbSymbBySymb}
	\end{align}
	
\end{proposition}

\begin{remark}[{\em Asymmetric channel}] \label{rem:asymmetry}
	The first term on the right-hand-side (RHS) of \eqref{eq:errProbSymbBySymb} corresponds to the probability of error when $X=0$ is transmitted, while the second term corresponds to the case of $X = \Delta$. As we consider a non-negative and heavy-tailed distribution, it follows that:
\begin{equation*}
	1 - \erfc \left( \sqrt{\frac{c}{2\theta}} \right) \gg \erfc \left( \sqrt{\frac{c}{2(\theta - \Delta)}} \right).
\end{equation*}	
	
	\noindent This implies that the channel is asymmetric, and the probabilities of error in sending $S=0$ or $S=1$ are different. 
	\changeYony{The probabilities of error for each of the symbols can be made equal} by alternating the assignments of bits in \eqref{eq:TxMapping} over time, or by applying coding dedicated to asymmetric channels, see \cite{Mondelli_arXiv14, constantin, zhou13} and references therein. 
\end{remark}

\begin{proof}[Proof of Proposition \ref{prop:decRuleSymbBySymb}]
The optimal symbol-by-symbol decision rule is the MAP rule \cite[Ch. 4.1]{ProakisDigComm}. As we consider a binary detection problem with equiprobable constellation points, the MAP rule specializes to the ML rule, which using the mapping \eqref{eq:TxMapping} is written as:
\begin{align}
	\frac{f_{Y|X}(y|x=0)}{f_{Y|X}(y_1|x=\Delta)} \mspace{8mu} \begin{matrix} \hat{S} = 0 \\ \gtrless \\ \hat{S} = 1 \end{matrix} \mspace{8mu} 1, \quad y_1 > \Delta.
	\label{eq:MAPrule}
\end{align}

\noindent Plugging the density in \eqref{eq:LevyNoise} with $\mu = x$ into the left hand side (LHS) of \eqref{eq:MAPrule}, and applying $\log(\cdot)$ on both sides, we obtain \eqref{eq:thetaEquation}. The uniqueness of the threshold $\theta$ follows from the fact that the PDFs for both hypotheses are shifted versions of the L\'evy PDF, which is unimodal \cite[Ch. 2.7]{zolotarev-book}. A formal and rigorous proof for this uniqueness is provided in Appendix \ref{annex:Uniqueness_proof}.

Regarding the probability of error, for the case of $y_1<\Delta$, we note that due to causality $s$ must be equal to $0$. For $y_1 \ge \Delta$ we write:
	\begin{align*}
		P_{\varepsilon} & \stackrel{(a)}{=} 0.5 \left( \Pr \{ y > \theta | s=0 \} + \{ y \le \theta | s=1 \} \right) \\
		& \stackrel{(b)}{=} 0.5 \left(1 - \erfc \left( \sqrt{\frac{c}{2\theta}} \right) + \erfc \left( \sqrt{\frac{c}{2(\theta - \Delta)}} \right) \right),
 	\end{align*}
	
\noindent 	where (a) follows from the assumption that the symbols are equiprobable, and (b) follows from \eqref{eqn:LevyCDF}. 
We emphasize that this proposition can be easily extended to the case of unequal a-priori symbol probabilities. 
\end{proof}

	
The probability of error in molecular communication with optimal detection can be reduced by transmitting multiple information particles for each symbol \cite{sri12}, \cite{men12}, namely, using $M>1$ particles for each transmission.\footnote{As we assume that the transmitter and the receiver are perfectly synchronized, the best strategy is to simultaneously release $M$ molecules. Releasing the $M$ molecules at different times can only increase the ambiguity at the receiver and therefore increase the probability of error \cite[Sec IV.C]{sri12}.} 
\changeYony{In fact, in \cite{ITsubmission} we showed that the capacity of the DBMT channel scales at least poly-logarithmically with $M$. Yet, capacity analysis in general, including that of \cite{ITsubmission}, does not provide an analysis of the probability of error, nor does it provide decoding methods for practical modulations.}
In this section we first present the ML detector for the DBMT channel, and then discuss lower-complexity detection approaches.

\subsection{ML Detection for $M>1$}
Let $\yvec = \{ y_m\}_{m=1}^M$. The following proposition characterizes the ML detector based on the channel outputs $\yvec$:
\begin{proposition}
	\label{prop:MLdetectorMultiple}
	The decision rule that minimizes the probability of error for $M \ge 1$ is given by:
	\begin{align}
		\hat{S}_{\text{ML}}(\yvec) \mspace{-3mu} = \mspace{-3mu} \begin{cases} 1, & \forall y_m: y_m \mspace{-3mu} > \mspace{-3mu} \Delta, \text{ and } \\ & \mspace{10mu} \sum_{m=1}^{M}{\log \mspace{-2mu} \left( \mspace{-2mu} \frac{y_m-\Delta}{y_m} \mspace{-2mu} \right) \mspace{-3mu} + \mspace{-3mu} \frac{c \Delta}{3} \frac{1}{y_m(y_m - \Delta)}} \mspace{-3mu} \le \mspace{-3mu} 0 \\ 0, &  \text{otherwise}. \end{cases}
		\label{eq:decisionRuleMultiple}
	\end{align}
\end{proposition}

\begin{proof}
	The proof follows along the same lines as the proof of Prop. \ref{prop:decRuleSymbBySymb}. More precisely, as the a-priori probabilities are equal, the optimal detection rule is ML. Using Assumption \ref{assmp:indep} the joint conditional density of $\yvec$ is a product of the individual conditional densities, and applying $\log(\cdot)$ results in the condition $\sum_{m=1}^{M}{\log \left( \frac{y_m-\Delta}{y_m} \right) + \frac{c \Delta}{3} \frac{1}{y_m(y_m - \Delta)}} \le 0$. Finally, as the additive noise is positive, if $\exists y_m: y_m \le 0$, then $\hat{S}_{\text{ML}}(\yvec) = 0$.
\end{proof}

Although the above ML detector minimizes the probability of error, it lacks an exact performance analysis and is relatively complicated to compute; this in particular holds for the $\log(\cdot)$ operation \cite{paul09, gutierrez11, klinefelter15}.
In the following we denote the probability of error of the ML detector by $P_{\varepsilon, \text{ML}}$. 
In traditional wireless communication, the common approach for reducing the complexity of detection is to apply {\em linear signal processing} to the sequence $\yvec$. 
The complexity of such a receiver is significantly lower compared to that of the ML detector, and for an AWGN channel this approach is known to be optimal \cite[Ch. 3.3]{TV:05}. \changeYony{In fact, even in non-Gaussian problems such as transmission over a timing channel with drift \cite[Sec. IV.C.2]{sri12}, modeled by the additive IG noise (AIGN) channel, the performance with linear detection improves by releasing multiple particles per symbol versus releasing just a single particle}.

In the next subsection we argue that for the DBMT channel a linear receiver performs better when each symbol consists of a single particle release versus multiple particle releases. The sub-optimality of multiple particle releases versus a single particle release when linear signal processing is applied at the receiver, under channels with $\alpha$-stable additive noise was observed in \cite[Ch. 10.4.6]{nikias-book}. Yet, to the best of our knowledge, the analysis in the next sub-section is the first to rigorously show this effect.

\subsection{Linear Detection for $M>1$}

	In this subsection we consider linear detection of signals transmitted over an additive channel corrupted by $\alpha$-stable noise with characteristic exponent smaller than unity, namely, we use the channel model \eqref{eq:LevyChan}, with the minor change that $Z_m \sim \mathscr{S}(0,c,\beta,\alpha), \alpha < 1$. Thus, the results presented in this subsection also hold for the L\'evy-distributed noise. 
	Let $\{ w_m \}_{m=1}^M$, $w_m \mspace{-3mu} \in \mspace{-3mu} \realSet^{+}, \sum_{m=1}^M \mspace{-3mu} w_m \mspace{-3mu} = \mspace{-3mu} 1$ be a set of coefficients, and consider ML detection based on $Y_{\text{LIN}} \mspace{-3mu} \triangleq \mspace{-3mu} \sum_{m=1}^M \mspace{-3mu} w_m Y_m$:
	\begin{align} 
		\hat{X}_{\text{LIN}} = \argmax_{x \in \{0, \Delta\}} f_{Y_{\text{LIN}}|X}(y_{\text{LIN}}|X=x). \label{eq:lin_detector}
	\end{align}
	
	\noindent Let $P_{\varepsilon, \text{LIN}}$ denote the probability of error of the detector $\hat{X}_{\text{LIN}}$. We now have the following theorem:
	\begin{theorem} \thmlabel{thm:badlindet}
	The probability of error of the linear detector under multiple particle release ($M>1$) is higher than the probability of error of the detector with the decision rule \eqref{eq:decisionRule} under single particle release ($M=1$), namely, $P_{\varepsilon, \text{LIN}} \ge P_{\varepsilon}$, where $P_{\varepsilon}$ is given in \eqref{eq:errProbSymbBySymb}.		
	\end{theorem}
	
	\begin{proof}
		We show that given $X = x$, $Y_{\text{LIN}} \sim \mathscr{S}(x, c_{\text{LIN}},\alpha,\beta)$, with $c_{\text{LIN}} \ge c$. 
		Note that when $X=x$ is given then the $Y_m$'s are independent. Therefore, the characteristic function of $Y_{\text{LIN}}$, given $X=x$, is given by:
	\begin{align*}
		& \varphi_{Y_{\text{LIN}}|X=x}(t) \nonumber \\ 
		& \quad = \prod_{m=1}^M \exp \Big\{ j x w_m t \nonumber \\
		& \mspace{120mu} - |c w_m t|^\alpha (1 - j \beta \sgn(w_m t) \Phi(w_m t,\alpha)) \Big\}	\\
		& \quad \stackrel{(a)}{=} \prod_{m=1}^M \exp \left\{ j x w_m t - |c w_m t|^\alpha (1 - j \beta \sgn(t) \Phi(t,\alpha))\right\}	\\
		& \quad = \exp \left\{ \sum_{m=1}^M \left\{ j x w_m t - |c w_m t|^\alpha (1 - j \beta \sgn(t) \Phi(t,\alpha)) \right\} \right\} \\
		& \quad \stackrel{(b)}{=} \exp \left\{ j x t - \left( \sum_{m=1}^M c w_m ^\alpha \right) |t|^\alpha (1 - j \beta \sgn(t) \Phi(t,\alpha))  \right\} \\
		& \quad \stackrel{(c)}{=} \exp \left\{ j x t -  |c_{\text{LIN}}t|^\alpha (1 - j \beta \sgn(t) \Phi(t,\alpha))  \right\},
	\end{align*}
	
	
	\noindent where (a) follows from the fact that $w_m \mspace{-3mu} > \mspace{-3mu} 0$ and from the fact that $\Phi(t,\alpha)$ is independent of $t$, for $\alpha \mspace{-3mu} < \mspace{-3mu} 1$ ; (b) follows from the fact that $\sum_{m=1}^M w_m = 1$; and (c) follows by defining $c_{\text{LIN}} = c \cdot \left( \sum_{m=1}^M w_m^{\alpha} \right)^{\frac{1}{\alpha}}$.
	Therefore, given $X=x$, we have $Y_{\text{LIN}} \sim \mathscr{S}(x, c_{\text{LIN}},\alpha,\beta)$. 
	Since $w_m \le 1, m=1,2,\dots,M$, we have $\left( \sum_{m=1}^M w_m^{\alpha} \right)^{\frac{1}{\alpha}} \ge 1$, and therefore $c_{\text{LIN}} \ge c$. Finally, as $c$ is the dispersion of the distribution, and since stable distributions are unimodal \cite[Ch. 2.7]{zolotarev-book}, it follows that the probability of error increases with $c$. Therefore, we conclude that $P_{\varepsilon, \text{LIN}} \ge P_{\varepsilon}$. 
	\end{proof}

As the L\'evy distribution is a special case of the family $\mathscr{S}(0,c,\beta,\alpha), \alpha < 1$, we have the following corollary:
\begin{corollary} \label{cor:linDegradeLevy}
	In DBMT channels without flow and $M>1$, the linear detector has worse performance compared to the case of $M=1$.
\end{corollary}

The result of Corollary \ref{cor:linDegradeLevy} is demonstrated in Section \ref{sec:numRes}.

\begin{remark}[{\em Comparison to the AIGN channel}]
	The difference between the AIGN channel (or the AWGN channel) and the channel considered in this paper stems from the fact that for the AIGN, the (weighted) averaging associated with linear detection can decrease the noise variance, namely, the tails of the noise. On the other hand, in the case of the L\'evy distribution, averaging leads to a heavier tail, and therefore to a higher probability of error.
\end{remark}

\begin{remark}[{\em Decision delay}]
In order to implement the ML detector \eqref{eq:decisionRuleMultiple}, the receiver must wait for all particles to arrive  as all particle arrival times are used in the detection algorithm. However, as the L\'evy distribution has heavy tails, this may result in very long decision delays. In fact, the average decision delay of such a receiver will be infinite. 
\end{remark}

In the next section we present a simple detector that is based on the time associated with the first particle arrival. This detector requires a short reception interval (on the order of the single-particle case) and achieves performance very close to that achieved by the ML detector for small values of $M$.


\section{Transmission over the DBMT Channel for $M>1$: FA Detection} \label{subsec:FA}
The detector proposed in this section detects the transmitted symbol based only on the FA among the $M$ particles, namely, it waits for the first particle to arrive and then applies ML detection based on this arrival. 
In terms of complexity, the FA detector simply compares the first arrival to a threshold; this is in contrast to the complicated ML detector in \eqref{eq:decisionRuleMultiple}. 

Let $y_{\text{FA}} = \min \{ y_1,y_2,\dots,y_M\}$. In the sequel we show that the PDF of $Y_{\text{FA}}$ is more concentrated towards the release time than the original L\'evy distribution.
 The FA detector is presented in the following theorem:
\begin{theorem} \label{thm:FA_detector}
	The decision rule that minimizes the probability of error, based on $y_{\text{FA}}$, is given by:
	\begin{align}
		\hat{S}_{\text{FA}}(y_{\text{FA}}) = \begin{cases} 0, & y_{\text{FA}} < \theta_M \\ 1, &  y_{\text{FA}} \ge \theta_M, \end{cases}
		\label{eq:decisionRuleFA}
	\end{align}
	
	\noindent where $\Delta \le \theta_M \le \theta_{M-1}, \theta_1 = \theta$, is the solution of the following equation in $y_{\text{FA}}$:
	\begin{align}
		& y_{\text{FA}}(y_{\text{FA}} \mspace{-3mu} - \mspace{-3mu} \Delta) \mspace{-3mu} \cdot \mspace{-3mu}  \log \left( \frac{y_{\text{FA}}}{y_{\text{FA}} \mspace{-3mu} - \mspace{-3mu} \Delta} \right)  \nonumber \\
		& \mspace{10mu} + \mspace{-3mu} y_{\text{FA}}(y_{\text{FA}} \mspace{-3mu} - \mspace{-3mu} \Delta) \mspace{-3mu} \cdot \mspace{-3mu} \log \mspace{-3mu} \left( \mspace{-3mu} \frac{1 \mspace{-3mu} - \mspace{-3mu} \erfc \mspace{-3mu} \left( \sqrt{\frac{c}{2(y_{\text{FA}} - \Delta)}} \right)}{1 \mspace{-3mu} - \mspace{-3mu} \erfc \mspace{-3mu} \left( \sqrt{\frac{c}{2y_{\text{FA}}}} \right)} \mspace{-3mu} \right)^{\mspace{-10mu} \frac{2(M \mspace{-3mu} - \mspace{-3mu}1)}{3}} \mspace{-12mu} = \mspace{-3mu} \frac{c \Delta}{3}, 
		\label{eq:thetaMEquation}
	\end{align}
	
	\noindent for $y_{\text{FA}} \ge \Delta > 0$. 	
	The probability of error of the FA detector is given by:
	\begin{align}
		P_{\varepsilon, \text{FA}} & = 0.5 \Bigg( \left( 1 - \erfc \left( \sqrt{\frac{c}{2\theta_M}} \right) \right)^M \nonumber \\
		& \mspace{40mu} + 1 - \left(1 - \erfc \left( \sqrt{\frac{c}{2(\theta_M - \Delta)}} \right) \right)^M \Bigg).
		\label{eq:errProbSymbBySymbFA}
	\end{align}
\end{theorem}

\begin{proof}
	The detection rule that minimizes the probability of error is the ML detector based on $Y_{\text{FA}}$. This requires the PDF and CDF of $Y_{\text{FA}}$ given $X$.  
	Let $F_{Y|X}(y|x)$ denote the CDF of $y_m$ given $X$. Assumption \ref{assmp:indep} implies that given $X$, the channel outputs $Y_1,Y_2,\dots,Y_M$ are independent. Hence, using basic results from order statistics \cite[Ch. 2.1]{OrderStat-book}, we write:
	\begin{align}
		F_{Y_{\text{FA}}|X}(y_{\text{FA}}|x) & = 1 - \left( \Pr \{ Y > y | X=x \} \right)^M \nonumber \\
		& \stackrel{(a)}{=} 1 - \left(1 - \erfc \left( \sqrt{\frac{c}{2(y - x)}} \right) \right)^M \nonumber \\
		& \triangleq \Psi \left( c, M, y-x \right), 
		\label{eq:yFAcdf}
	\end{align}
	
	\noindent where (a) follows from \eqref{eqn:LevyCDF}.
	Next, to obtain the PDF of $Y_{\text{FA}}$ given $X$, we write:
	\begin{align}
		f_{Y_{\text{FA}}|X}(y_{\text{FA}}|x) & = \mspace{-3mu} \frac{\partial F_{Y_{\text{FA}}|X}(y_{\text{FA}}|x)}{ \partial y_{\text{FA}}} \nonumber \\
		& = \mspace{-3mu} M \mspace{-3mu} \cdot \mspace{-3mu} f_{Y|X}(y|x) \mspace{-3mu} \cdot \mspace{-3mu} \left( \mspace{-3mu} 1 \mspace{-3mu} - \mspace{-3mu} \erfc \left( \sqrt{\frac{c}{2(y \mspace{-2mu} - \mspace{-2mu} x)}} \mspace{-3mu} \right) \mspace{-3mu} \right)^{\mspace{-3mu} M-1} \nonumber \\
		& = \mspace{-3mu} M \cdot \sqrt{\frac{c}{2 \pi (y-x)^3}} \exp \left( - \frac{c}{2(y-x)} \right) \nonumber \\
		& \mspace{45mu} \times \mspace{-3mu} \left(1 \mspace{-3mu} - \mspace{-3mu} \erfc \left( \sqrt{\frac{c}{2(y - x)}} \right) \right)^{M-1}. 
		\label{eq:yFApdf}
	\end{align}
	
\noindent Hence, the ML decision rule based on the measurement $y_{\text{FA}}$ is given by:
\begin{align}
	\frac{f_{Y_{\text{FA}}|X}(y_{\text{FA}}|x=0)}{f_{Y_{\text{FA}}|X}(y_{\text{FA}}|x=\Delta)} \mspace{8mu} \begin{matrix} \hat{S} = 0 \\ \gtrless \\ \hat{S} = 1 \end{matrix} \mspace{8mu} 1, \quad y_{\text{FA}} > \Delta.
	\label{eq:MAPruleFA}
\end{align}

\noindent Plugging the density in \eqref{eq:yFApdf} into the LHS of \eqref{eq:MAPruleFA}, and applying some algebraic manipulations we obtain \eqref{eq:thetaMEquation}. 

To show that $\theta_M \le \theta_{M-1}$ we first note that by plugging \eqref{eq:yFApdf} into \eqref{eq:MAPruleFA} it follows that $\theta_M$ is the solution of the following equation:
\begin{align}
	\frac{f_{Y|X}(y_{\text{FA}}|x=0)}{f_{Y|X}(y_{\text{FA}}|x=\Delta)} \mspace{-3mu} = \mspace{-3mu} \left(\frac{1 - \erfc \left( \sqrt{\frac{c}{2 (y_{\text{FA}} - \Delta)}}  \right) }{1 - \erfc \left( \sqrt{\frac{c}{2y_{\text{FA}} }} \right)} \right)^{\mspace{-5mu} M-1}. \label{eq:equibal_ofM}
\end{align}

\noindent Now, for $M=1$, the RHS of \eqref{eq:equibal_ofM} equals 1, and $\theta_1 \in [\Delta, \Delta + \frac{c}{3}]$. Thus, in this interval, the LHS of \eqref{eq:equibal_ofM} achieves the value 1. An explicit evaluation of the derivative of the LHS of \eqref{eq:equibal_ofM} shows that in this range the derivative is negative, and therefore the LHS of \eqref{eq:equibal_ofM} decreases with $y_{\text{FA}}$, independently of $M$. On the other hand, the RHS of \eqref{eq:equibal_ofM} increases with $M$ for all $y_{\text{FA}} \ge \Delta$. Therefore, we conclude that the solution of \eqref{eq:equibal_ofM} decreases with $M$.

Regarding the probability of error, we first note that for $y_{\text{FA}}<\Delta$, due to the causality of the arrival time, $S$ must be equal to $0$, and therefore the probability of error is zero. For $y_{\text{FA}} \ge \Delta$ we write:
\begin{align*}
		P_{\varepsilon, \text{FA}} & = 0.5 \Big( 1 - F_{Y_{\text{FA}}|X}(y_{\text{FA}}|x=\theta_M) \nonumber \\
		 & \mspace{65mu} + F_{Y_{\text{FA}}|X}(y_{\text{FA}}|x = \theta_M - \Delta) \Big).
 	\end{align*}
	
\noindent By plugging the CDF in \eqref{eq:yFAcdf} into this expression we obtain \eqref{eq:errProbSymbBySymbFA}. Finally, we note that this theorem can also be easily extended to the case of unequal a-priori symbol probabilities.
\end{proof}

\begin{example}
	Consider sending information particles with diffusion coefficient $D=10 \mu m^2 /s$, see \cite{ber-book}, and let the distance between the transmitter and the receiver be $d=4\sqrt{10} \mu m$. This implies that $c = 2 s$. We further set $\Delta = 1$, and using Prop. \ref{prop:decRuleSymbBySymb}, for $M=1$, we obtain the optimal decision threshold $\theta = 1.372$. The conditional probability densities $f_{Y|X}(y|x=0)$ and $f_{Y|X}(y|x=\Delta)$ are illustrated in Fig. \ref{fig:CondDistributions}. 
	\changeYony{Fig. \ref{fig:CondDistributions} also depicts the conditional probability distributions for $M=3$ and $M=15$. For these cases the optimal decision thresholds are $\theta_3 = 1.286$ and $\theta_{15} = 1.146$. It can be observed that when $M$ is increased the conditional PDFs concentrate towards $X=0$ and $X=\Delta$. Moreover, the tails of the conditional PDFs in the case of $M=15$ are significantly smaller than the tails in the case of $M=3$ and $M=1$. 
	Finally, note that while the tail decreases exponentially in $M$, the PDF around $X=0$ or $X=\Delta$ increases linearly with $M$, see \eqref{eq:yFApdf}.}
		\renewcommand{\figWidth}{0.9}
		\begin{figure}
    \centering
    \includegraphics[width=\figWidth\columnwidth]{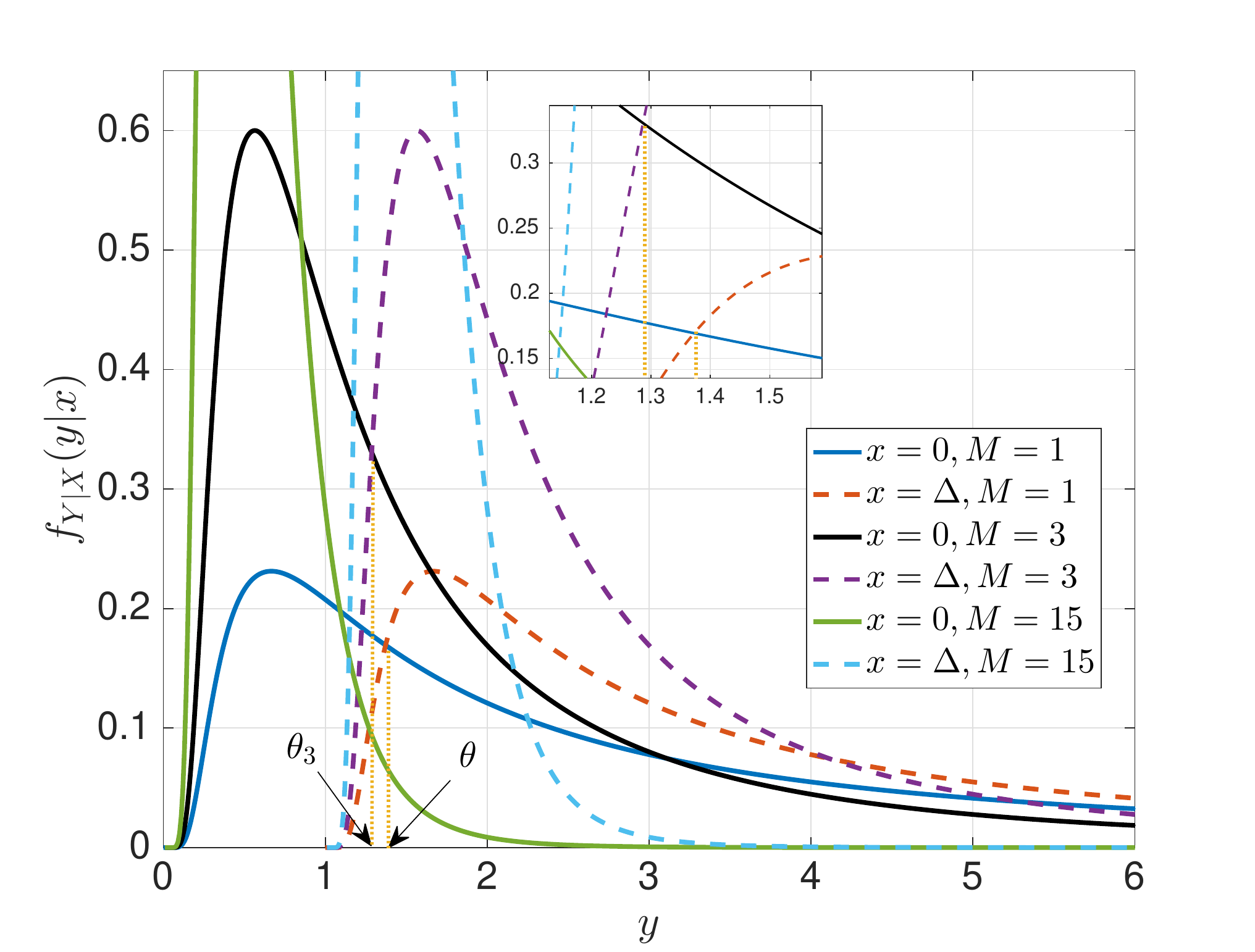}
    \captionsetup{font=footnotesize}
		\caption{The conditional probability densities $f_{Y|X}(y|x=0)$ and $f_{Y|X}(y|x=\Delta)$, for $c=2$, and $\Delta = 1$.}
    \label{fig:CondDistributions}
\end{figure}

\end{example}

\begin{remark}[{\em FA detector for $L>2$}]
	The FA detection framework can be directly extended to the case of $L>2$. In such cases the detection will be based on $L-1$ thresholds, which define the $L$ constellation points. Furthermore, as the conditional PDFs concentrate near $x$ when $M$ increases, we conclude that by increasing $M$ one can support larger $L$ for a given target probability of error.\footnote{Note that when $L>2$, then $\Pr \{S \ne \hat{S} \}$ refers to the {\em symbol} error probability.} This is demonstrated in Section \ref{sec:numRes}.  
\end{remark}

In contrast to \Thmref{thm:badlindet}, which states that for DBMT channels {\em without} drift, linear processing has worse performance for multiple particle release ($M>1$) versus single particle release ($M=1$), in \cite[Sec. IV.C.2]{sri12} it is shown that linear processing yields significant performance gain for DBMT channels {\em with} drift (modeled by the AIGN channel). Thus, a natural question that arises is: {\em Does the FA detector improve upon the linear detector for AIGN channels?} In the next subsection, and in the numerical results section, we show that this is the case.

\subsection{FA detection for the AIGN Channel}

Before comparing the FA detection framework and the linear processing proposed in \cite[Sec. IV.C.2]{sri12}, we briefly introduce the IG distribution. For a detailed discussion regarding the IG distribution we refer the reader to \cite{sri12} and \cite{chhikara-folks}. 
Consider a fluid medium {\em with drift velocity} $v$. Similarly to Section \ref{subsec:DBMTdef}, let $D$ denote the diffusion coefficient and $d$ denote the distance between the transmitter and the receiver. Moreover, let $\kappa \triangleq \frac{d}{v}$ and $\lambda \triangleq \frac{d^2}{2D}$. In this case the additive noise $Z_m$ in \eqref{eq:LevyChan} follows an inverse Gaussian distribution, $Z_m \sim \IG(\kappa, \lambda)$. The conditional PDF of the AIGN channel output $Y$, given channel input $X=x$, is given in \cite[eq. (7)]{sri12} as:
\begin{align}
	f^{\text{IG}}_{Y|X}(y|x) \mspace{-3mu} = \mspace{-3mu} \begin{cases} \sqrt{\frac{\lambda}{2 \pi (y - x)^3}} \exp \left( - \frac{\lambda (y - x - \kappa)}{2 \kappa^2 (y-x)} \right), & y \mspace{-3mu} > \mspace{-3mu} x \\ 0, & y \mspace{-3mu} \le \mspace{-3mu} x, \end{cases} \label{eq:IGcondPDF}
\end{align}

\noindent while the conditional CDF is given in \cite[eq. (22)]{sri12} as:
\begin{align}
	F^{\text{IG}}_{Y|X}(y|x) = \begin{cases} \PhiG \left( \sqrt{\frac{\lambda}{y-x}} \left(\frac{y-x}{\kappa} \mspace{-3mu} - \mspace{-3mu} 1 \right) \right) \\ \quad + \mspace{3mu} e^{\frac{2\lambda}{\kappa}} \PhiG \left( - \sqrt{\frac{\lambda}{y-x}} \left(\frac{y-x}{\kappa} \mspace{-3mu} + \mspace{-3mu} 1 \right) \right), & y \mspace{-3mu} > \mspace{-3mu} x \\ 0, & y \mspace{-3mu} \le \mspace{-3mu} x. \end{cases} \label{eq:IGcondCDF}
\end{align}

\noindent In \cite[Sec. IV.C.2]{sri12} the authors proposed to average the $M$ channel outputs as:
\begin{align}
	Y_{\text{LIN}} = \frac{1}{M} \sum_{m=1}^M Y_m \stackrel{(a)}{=} X + Z_{\text{LIN}}, \label{eq:avgIG}
\end{align}

\noindent where (a) follows by defining $Z_{\text{LIN}} \triangleq \frac{1}{M} \sum_{m=1}^{M} Z_m$, where $Z_{\text{LIN}} \sim \IG(\kappa, M \cdot \lambda)$. Then, $\hat{X}_{\text{LIN}}$ is detected from $y_{\text{LIN}}$ as in \eqref{eq:lin_detector}. 
Note that the variance of an IG-distributed RV is given by $\frac{\kappa^3}{\lambda}$. Therefore, compared to the single particle case, the averaging in \eqref{eq:avgIG} decreases the variance by a factor of $M$, which partially explains the performance improvement, compared to the case of $M=1$, reported in \cite[Fig. 9]{sri12}. 

Let $f^{\text{IG}}_{Y_{\text{FA}}|X}(y_{\text{FA}}|x=0)$ denote the conditional PDF of $Y_{\text{FA}}$, given $X=0$. 
This PDF can be obtained by following the steps leading to \eqref{eq:yFApdf}, and using the PDF and CDF of the IG distribution given in \eqref{eq:IGcondPDF} and \eqref{eq:IGcondCDF}, respectively. 
To qualitatively compare the FA detector and the linear detector presented in \cite[Sec. IV.C.2]{sri12}, in the case of the AIGN channel, we propose to examine the tails of $f^{\text{IG}}_{Y_{\text{LIN}}|X}(y_{\text{LIN}}|X=0)$ and $f^{\text{IG}}_{Y_{\text{FA}}|X}(y_{\text{FA}}|X=0)$. 
Both PDFs are more concentrated around $X=0$ than $f^{\text{IG}}_{Y_m|X}(y|x=0)$. While in the case of the linear detector this is a result of the lower variance, in the case of FA this is a result of the multiplication by the exponential term, see \eqref{eq:yFApdf}. Our analysis shows that $f^{\text{IG}}_{Y_{\text{FA}}|X}(y_{\text{FA}}|x=0)$ is more concentrated around $X = 0$ than $f^{\text{IG}}_{Y_{\text{LIN}}|X}(y_{\text{LIN}}|X=0)$, as indicated in the following example.

\begin{example}
To obtain some intuition why the FA detector improves upon the linear detector in \eqref{eq:avgIG}, Fig. \ref{fig:CondDistributionsIG} depicts $f^{\text{IG}}_{Y_m|X}(y|x=0), f^{\text{IG}}_{Y_{\text{LIN}}|X}(y_{\text{LIN}}|x=0)$, and $f^{\text{IG}}_{Y_{\text{FA}}|X}(y_{\text{FA}}|x=0)$, for $\gamma = 1, \kappa = 1$ and $M=4$. It can be observed that $f^{\text{IG}}_{Y_{\text{FA}}|X}(y_{\text{FA}}|x=0)$ is significantly more concentrated towards the origin compared to $f^{\text{IG}}_{Y_{\text{LIN}}|X}(y_{\text{LIN}}|x=0)$. Thus, as the detector compares two shifted versions of the same PDF, this leads to a lower probability of error. 
It can further be observed that the tail of $f^{\text{IG}}_{Y_{\text{LIN}}|X}(y_{\text{LIN}}|x=0)$ is smaller than the tail of $f^{\text{IG}}_{Y_m|X}(y|x=0)$ which is reflected in the smaller variance. This supports the performance gain of the linear detector compared to ML detection for $M=1$.
		\renewcommand{\figWidth}{0.9}
\begin{figure}
    \centering
    \includegraphics[width=\figWidth\columnwidth]{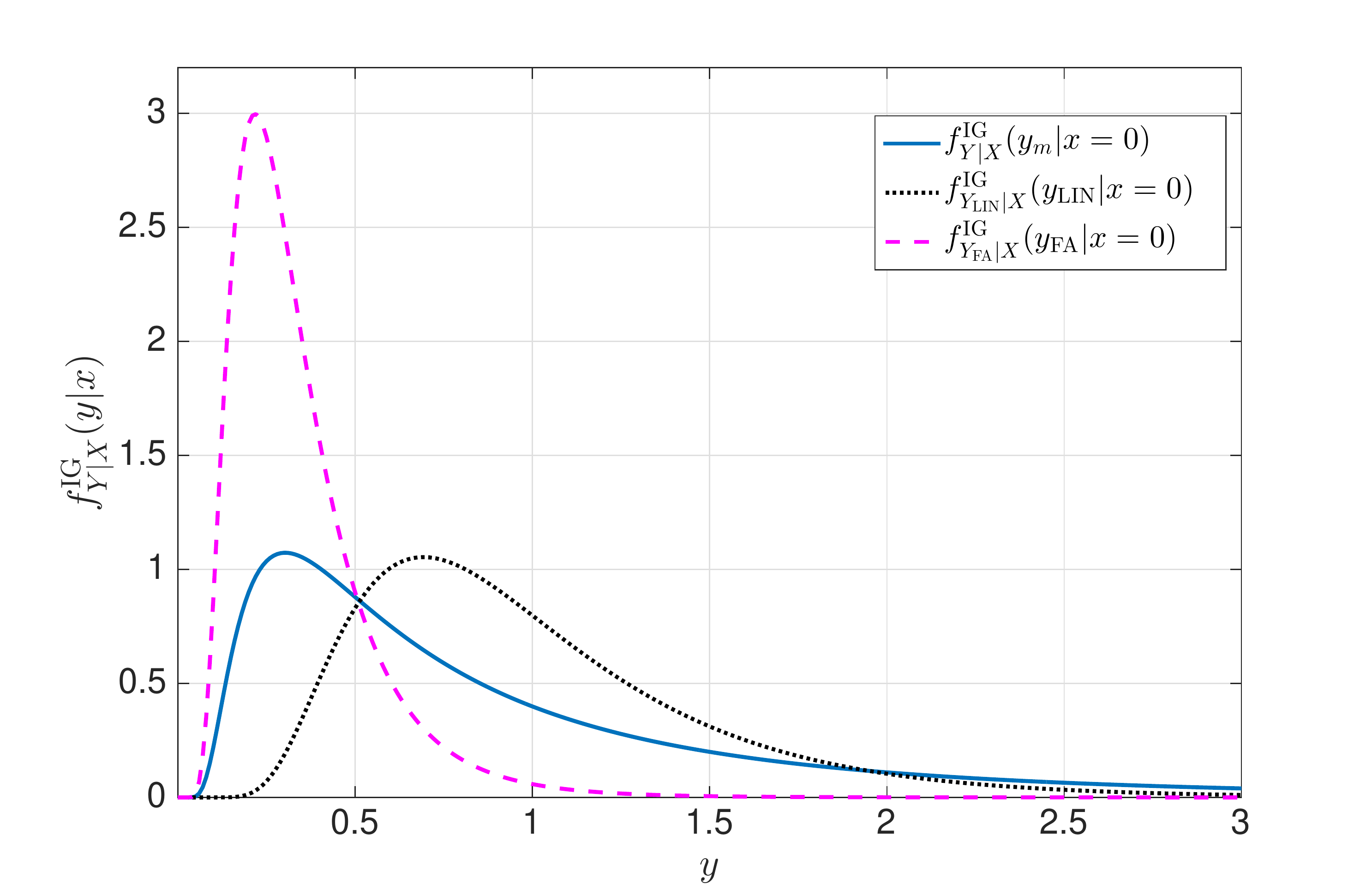}
    \captionsetup{font=footnotesize}
		\caption{DBMT channel with drift - The conditional probability densities $f^{\text{IG}}_{Y_m|X}(y|x=0), f^{\text{IG}}_{Y_{\text{LIN}}|X}(y_{\text{LIN}}|x=0)$, and $f^{\text{IG}}_{Y_{\text{FA}}|X}(y_{\text{FA}}|x=0)$, for $\lambda = 1, \kappa = 1$ and $M=4$.}
    \label{fig:CondDistributionsIG}
\end{figure} 
\end{example}

Finally, we note that while Fig. \ref{fig:CondDistributionsIG} provides a qualitative explanation for the superiority of the FA detector compared to the linear detector, in Fig. \ref{fig:IG_PeVsVelocity_DifMs} we provide simulation results which support this observation.

\section{Performance Comparison of the ML and FA Detectors} \label{subsec:PerfCompare}

\changeYony{In this section we compare the probability of error of the FA detector to the probability of error of the ML detector. Clearly, $Y_{\text{FA}}$ is {\em not} a sufficient statistic for decoding based on $\yvec$, yet, our numerical simulations indicate that for low values of $M$ (up to the order of tens), these detectors have almost equivalent performance. On the other hand, when $M$ is large (i.e., $M \to \infty$), we use error exponent analysis to show the superiority of the ML detector over the FA detector.}

\subsection{Small $M$}
	To study the performance gap between the two detectors, when $M$ is small, we derive an upper bound on the probability that there is a mismatch between the decisions of the two detectors. More precisely, let $P_{\text{mm}} \triangleq \Pr \{\hat{S}_{\text{ML}}(\yvec) \neq \hat{S}_{\text{FA}}(\yvec)\}$. The following theorem upper bounds $P_{\text{mm}}$:
	\begin{theorem} \thmlabel{thm:PmmUB}
		Let $g(x) \triangleq \log \left( \frac{x - \Delta}{x} \right) + \frac{c\Delta}{3 x (x-\Delta)}, x > \Delta$. The equation $g(x)=0$ has a unique solution $x^{\ast}$, where $g(x) > 0$ for $\Delta < x < x^{\ast}$, and $g(x) < 0$ for $x^{\ast} < x$. Furthermore, the mismatch probability is upper bounded by:
	\begin{align}
		P_{\text{mm}} \le P_{\text{mm}}^{\text{(ub)}} & = 0.5 \sum_{i=0}^1 \Big\{ \Psi \left( c, M, x^{\ast} - i\cdot \Delta \right) \nonumber \\
			& \mspace{110mu} - \Psi \left( c, M, \bar{i} \Delta \right) \Big\}, \label{eq:Pmm_ub}
	\end{align}
	
	\noindent where $\bar{0} = 1$, and $\bar{1} = 0$.
\end{theorem}

\begin{proof}
	The proof is provided in Appendix \ref{annex:thm_PmmUB_proof}.
\end{proof}

\begin{remark}[{\em Tightness of the bound in \eqref{eq:Pmm_ub}}]
	Recall that $f_{Y_{\text{FA}}|X}(y|x=x_0)$ concentrates towards $x=x_0$ when $M$ increases. On the other hand, $x^{\ast}$ and $\Delta$ are independent of $M$ and depend on the propagation of a {\em single} particle. Therefore, when $M$ increases, the upper bound in \eqref{eq:Pmm_ub} becomes loose. 
	We further note that the upper bound in \eqref{eq:Pmm_ub} is tightened when $\Delta$ is increased.
	For instance, let $M=2, c=1$, and $\Delta = 1$. For this setting $P_{\varepsilon,\text{ML}} = 0.2174, P_{\varepsilon,\text{FA}} = 0.2186$, and $P_{\text{mm}}^{\text{(ub)}} = 0.0283$. If we increase $\Delta$ to be equal to $5$ we obtain: $P_{\varepsilon,\text{ML}} = 0.05896, P_{\varepsilon,\text{FA}} = 0.05898$, and $P_{\text{mm}}^{\text{(ub)}} = 0.0012$. 
	On the other hand, for larger values of $M$, e.g., $M=5$, we have $P_{\varepsilon,\text{ML}} = 0.06501, P_{\varepsilon,\text{FA}} = 0.06554, P_{\text{mm}}^{\text{(ub)}} = 0.0337$, for $\Delta = 1$, and $P_{\varepsilon,\text{ML}} = 0.002403, P_{\varepsilon,\text{FA}} = 0.002408, P_{\text{mm}}^{\text{(ub)}} = 0.001$, for $\Delta = 5$.
\end{remark}

For large values of $M$, we next analyze the error exponents of the FA and ML detectors, and show that in this regime the ML detector significantly outperforms the FA detector.

\subsection{Large $M$} \label{subsec:LargeM}

Let $P_{\varepsilon}^{(M)}$ denote the probability of error of a given detector, as a function of $M$. The error exponent is then given by:
\begin{align}
	\mathsf{E} = \lim_{M \to \infty} - \frac{\log P_{\varepsilon}^{(M)}}{M}. \label{eq:errExpDef}
\end{align}

\changeYony{
\begin{remark}
	From \Thmref{thm:badlindet} it follows that for the linear detector the probability of error {\em does not} decrease when $M$ is increased, namely, $P_{\varepsilon, \text{LIN}}^{(M)} \ge P_{\varepsilon}$. Therefore, the definition of the error exponent in \eqref{eq:errExpDef} implies that $\mathsf{E}_{\text{LIN}} = 0$.
\end{remark}
}

In the following we first derive the error exponent of the FA detector, and then numerically compare it to the exponent of the ML detector. This numerical comparison indicates that the error exponent of the ML detector is higher than that of the FA detector. This implies that the two detectors are not equivalent, even though for low values of $M$ they achieve very similar performance based on our simulation results. This performance gap is due to the fact that the first arrival {\em is not a sufficient statistic} for optimal decoding based on the received vector $\yvec$.
The following theorem presents the error exponent of the FA detector:
\begin{theorem} \label{thm:ErrExp_FA}
	The error exponent of the FA detector is given by:
	\begin{align}
		\mathsf{E}_{\text{FA}} = - \log \left( 1 - \erfc \left( \sqrt{\frac{c}{2 \Delta}} \right) \right). \label{eq:FA_errExp}
	\end{align}
\end{theorem}

\begin{proof}[Proof Outline]
	Recall the probability of error of the FA detector in \eqref{eq:errProbSymbBySymbFA}, repeated here for ease of reference:
	\begin{align*}
		P_{\varepsilon, \text{FA}} &= 0.5 \bigg( \left( 1 - \erfc \left( \sqrt{\frac{c}{2\theta_M}} \right) \right)^M \nonumber \\
		& \mspace{65mu} + 1 - \left(1 - \erfc \left( \sqrt{\frac{c}{2(\theta_M - \Delta)}} \right) \right)^M \bigg).
	\end{align*}
	
	\noindent Based on the observations in Remark \ref{rem:asymmetry}, and noting that both PDFs are right-sided, namely, different than zero only for $y>x$, we intuitively expect the first term on the RHS of \eqref{eq:errProbSymbBySymbFA} to be larger than the second term. 
	In this case, the error exponent of the FA detector is governed by the first term, and as $\theta_M \to \Delta$ when $M \to \infty$, we obtain \eqref{eq:FA_errExp}. 
	In Appendix \ref{annex:thm_ErrExp_FA_proof} we rigorously analyze the scaling behavior of the second term in \eqref{eq:FA_errExp} and show that it yields {\em the same error exponent as the first term}, thus, leading to \eqref{eq:FA_errExp}.
\end{proof}

Next, we discuss the error exponent of the ML detector. Deriving a closed form expression for this error exponent seems intractable, therefore, we present an implicit expression and evaluate it numerically.
The problem of recovering $x$ based on the $M$ i.i.d. realizations $\{y_m\}_{m=1}^M$ belongs to the class of binary hypothesis problems, which are studied in \cite[Ch. 11]{cover-book}. In particular, the error exponent for the probability of error is exactly the Chernoff information \cite[Theorem 11.9.1]{cover-book}. 
We emphasize that this optimal error exponent is independent of the prior probabilities associated with the two values of the transmitted symbol $x$, see the discussion in \cite[pg. 388]{cover-book}. Thus, the assumption of equiprobable bits places no limitation on the error exponent of the ML detector.

Let $\pi_0>0$ and $\pi_{\Delta}>0$ denote a-prior probabilities for sending $x=0$ and $x=\Delta$, respectively, for a fixed $M$. Furthermore, let $g_{0}(y)$ and $g_{\Delta}(y)$ denote the likelihood functions corresponding to $x=0$ and $x=\Delta$, respectively. Finally, let $\mathtt{I}(\text{``condition"})$ denote the indicator function which takes the value 1 if the ``condition" is satisfied and zero otherwise.
Since given $x$ the $\{y_m\}_{m=1}^M$ are independent, it follows that the probability of error of the ML detector, as a function of $M$, can 	be written as:
\begin{align}
 P_{\varepsilon, \text{ML}}^{(M)} & \mspace{-3mu} = \mspace{-3mu}  \pi_0 \mspace{-3mu} \int_{\mathbf{y}} \prod_{m=1}^{M} \mspace{-3mu} g_0(y_m) \mathtt{I} \mspace{-3mu} \left(\prod_{m=1}^M \mspace{-6mu} g_0(y_m) \mspace{-3mu} < \mspace{-3mu} \prod_{m=1}^M \mspace{-6mu} g_{\Delta} (y_m)\right) d\mathbf{y} \nonumber \\
& \mspace{8mu} + \mspace{-3mu} \pi_{\Delta} \mspace{-6mu} \int_{\mathbf{y}} \prod_{m=1}^M \mspace{-3mu} g_{\Delta}(y_m) \mathtt{I} \mspace{-3mu} \left(\prod_{m=1}^M \mspace{-6mu} g_0(y_m) \mspace{-3mu} > \mspace{-3mu} \prod_{m=1}^M \mspace{-6mu} g_{\Delta} (y_m) \mspace{-3mu} \right) \mspace{-3mu} d\mathbf{y}. \label{eq:Pe_ML_ErrExp}
\end{align}

\noindent Next, we define:
\begin{equation*}
	\mathtt{J}_{M} \triangleq  \int_{\mathbf{y}} \min \left\{ \prod_{m=1}^M g_0(y_m), \prod_{m=1}^M g_{\Delta}(y_m) \right\} d\mathbf{y},
\end{equation*}
 
\noindent and note that \eqref{eq:Pe_ML_ErrExp} satisfies:
\begin{align}
  \label{eq:1}
\min \left\{ \pi_0, \pi_{\Delta} \right\} \mathtt{J}_M \leq P_{\varepsilon, \text{ML}}^{(M)} \leq \max \left\{\pi_0, \pi_{\Delta} \right\} \mathtt{J}_M.
\end{align}

\noindent Observe that for fixed $\pi_0$ and $\pi_{\Delta}$, the error exponent $\lim_{M\rightarrow \infty} \frac{-\log(P_{\varepsilon,\text{ML}}^{(M)})}{M}$ equals the error exponent of $\mathtt{J}_M$, namely,
\begin{align*}
	\lim_{M\rightarrow \infty} \frac{-\log(P_{\varepsilon, \text{ML}}^{(M)})}{M} = \lim_{M\rightarrow \infty} \frac{-\log(\mathtt{J}_M)}{M},
\end{align*} 

\noindent which is exactly the {\em Chernoff information}, see \cite[pg. 387]{cover-book}. We further write:
\begin{align*}
  \mathtt{J}_M & =  \int_{\mathbf{y}} \min \left\{ \prod_{m=1}^M g_0(y_m), \prod_{m=1}^M g_{\Delta}(y_m) \right\} d\mathbf{y} \\
	& \stackrel{(a)}{\leq} \min_{s:0 \leq s \leq 1} \int_{\mathbf{y}} \left( \prod_{m=1}^M g_0(y_m) \right)^{s} \left( \prod_{m=1}^M g_{\Delta}(y_m) \right)^{(1-s)} d\mathbf{y} \\
	& \stackrel{(b)}{\leq} e^{-M \mathsf{E}_{\text{ML}}}, 
\end{align*}

\noindent where (a) follows from the fact that for any positive numbers $a,b$ and a real number $s \in [0, 1]$, we have $ \min(a, b) \leq a^sb^{1-s}$, and (b) follows from defining:
\begin{align}
	\mathsf{E}_{\text{ML}} \triangleq - \min_{s: 0 \leq s \leq 1} \log\left(\int_{y} g^{s}_0(y) g^{1-s}_{\Delta}(y) dy\right). \label{eq:ML_errExp}
\end{align}

\changeYony{
\noindent The above argument establishes an upper  bound on the error exponent $\lim_{M \rightarrow \infty} \frac{-\log(\mathtt{J}_M)}{M}.$  A lower bound follows directly from \cite[Theorem 11.9.1]{cover-book}, namely, the best achievable exponent. The two bounds coincide as $M \rightarrow \infty$, see the discussion in \cite[pgs. 387--389]{cover-book}. Thus, we conclude that the error exponent of the ML detector is given in \eqref{eq:ML_errExp}.
}

\begin{example}
		In contrast to \eqref{eq:FA_errExp}, deriving a closed form expression for the error exponent \eqref{eq:ML_errExp} seems intractable, hence, we numerically evaluated it. Table \ref{tab:ErrExp} details both $\mathsf{E}_{\text{ML}}$ and $\mathsf{E}_{\text{FA}}$ for $\Delta \in \{0.1, 0.2, 0.3, 0.4\}$, and $c \in \{0.5, 1, 2\}$. 
		Note that when $M$ increases, very small values of $\Delta$ can be used. For instance, for $M=2\cdot 10^4, c=2$ and $\Delta=0.1$, we obtain $P_{\varepsilon,\text{FA}} = 2\cdot 10^{-4}$.
		It can be observed that for small values of $\Delta$, and large values of $c$, the relative difference between the two error exponents is larger. 
		\begin{table}[h]
		\begin{center}
		\footnotesize
		\begin{tabular}[t]{|c|c|c|c|c|}
			\hline
			  & $\Delta=0.1$  & $\Delta=0.2$ & $\Delta=0.3$ & $\Delta=0.4$  \\
			\hline
			\hline
			$ c= 0.5, \mathsf{E}_{\text{ML}}$  & 0.044106  & 0.132051  & 0.223149  & 0.306514  \\
			\hline
			$ c= 0.5, \mathsf{E}_{\text{FA}}$  & 0.025674  & 0.120865  & 0.219034  & 0.305917  \\
			\hline 
			\hline
			$ c= 1, \mathsf{E}_{\text{ML}}$  & 0.012413  & 0.044103  & 0.086111  & 0.132012  \\
			\hline
			$ c= 1, \mathsf{E}_{\text{FA}}$  & 0.001567  & 0.025674  & 0.070304  & 0.120865  \\
			\hline 
			\hline
			$ c= 2, \mathsf{E}_{\text{ML}}$  & 0.003230  & 0.012413  & 0.026441  & 0.044099  \\
			\hline
			$ c= 2, \mathsf{E}_{\text{FA}}$  & 0.000008  & 0.001567  & 0.009872  & 0.025674  \\
			\hline
		\end{tabular}
		\captionsetup{font=small}
		\caption{$\mathsf{E}_{\text{ML}}$ and $\mathsf{E}_{\text{FA}}$ for different values of $\Delta$ and $c$. \label{tab:ErrExp}}
		\vspace{-0.5cm}
	\end{center}
\end{table}

\end{example}

\section{Non-Binary Constellations} \label{sec:beyondBinary}

In this section we study communication over DBMT channels when $|\mX| = 2^L, L > 1$. We restrict our attention to the FA detection framework due to the complexity of the ML analysis. 
Let $L$ be a fixed number of bits to be transmitted, $\Delta$ a fixed time interval, and $\{\xi_i \}_{i=0}^{2^L - 1}$ a set of distinct points in the interval $[0, \Delta]$.
One can send the $L$ bits by releasing the $M$ particles at one of $\xi_i$ time points. 
The results of Section \ref{subsec:FA}, indicating that simultaneous release of multiple particles can dramatically decrease the probability of error in the binary case, also apply in this non-binary case. Therefore, for a fixed $L$, one can achieve a desired (symbol) probability of error by increasing the number of released particles. 
On the other hand, for a fixed $M$, increasing $L$ increases the number of bits conveyed in each symbol at the cost of smaller spacing between the $\xi_i$'s. This leads to two questions associated with the non-binary case:
\begin{itemize}
	\item 
		{\em What is the complexity of the FA detection in the case of $L>1$? Does it grow exponentially with $L$?} 
		\changeYony{We show that given a simple choice of the points $\{\xi_i \}_{i=0}^{2^L - 1}$, the FA detector for the case of $L>1$ amounts to the FA detector presented in Thm. \ref{thm:FA_detector}.}
				
	\item
		{\em What is the scaling behavior of $L$ as a function of $M$, which insures a decreasing symbol error probability?} We show that if $L$ scales at most as $\log \log M$ then the symbol error probability decreases to zero when $M, L \to \infty$.
	
\end{itemize}

We begin with formally introducing the transmission scheme. 
The transmitter divides the interval $[0,\Delta]$ into $2^L - 1$ equal-length sub-intervals. Let $\tilde{\Delta}$ be the length of each such sub-interval. The constellation points (release times) are given by $n \cdot \tilde{\Delta}, n=0,1,\dots,2^L-1$.
Observing a sequence of $L$ equiprobable and independent bits, the transmitter uses a predefined bits-to-symbol mapping and releases $M$ particles at the corresponding time. While in this work we focus on the {\em symbol} error rate, the bit error rate can be easily obtained via a bits-to-symbol mapping such as Gray coding \cite{agrell2004} and the approximation that a symbol error leads to a single bit error. The transmission scheme is illustrated in Fig. \ref{fig:BeyondBinarry}.
		\renewcommand{\figWidth}{0.9}
\begin{figure}[t]
	\begin{center}
		\includegraphics[width=\figWidth\columnwidth,keepaspectratio]{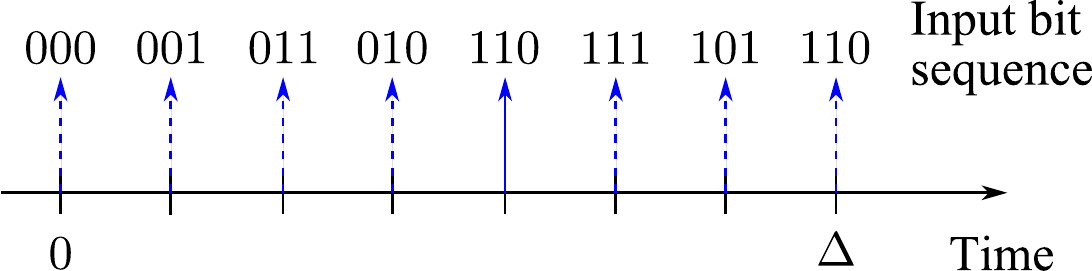}
	\end{center}
	\vspace{-0.3cm}
	\captionsetup{font=footnotesize}
	\caption{\label{fig:BeyondBinarry} Illustration of transmission when $L=3$. The bits-to-symbol mapping is the commonly-used Gray coding.
	Assuming the binary input sequence $110$ the transmitter releases $M$ particles at time $4\tilde{\Delta} = \frac{4\Delta}{7}$. The dashed arrows indicate other possible release times and the respective bit tuples.} 
	\vspace{-0.35cm}
\end{figure}

Let $\omega_M$ denote the mode of the density of the FA, given in \eqref{eq:yFApdf}, assuming the offset parameter is zero.\footnote{Note that the unimodality of the density in \eqref{eq:yFApdf} follows from the fact that the density of the L\'evy distribution is unimodal.}
Furthermore, let $n_{\text{floor}}(y_{\text{FA}}) \triangleq \left\lfloor \frac{(y_{\text{FA}} - \omega_M)}{\tilde{\Delta}} \right\rfloor$ and let $\hat{S}_{\text{FA}}^{\Delta^{\ast}}(y_{\text{FA}})$ denote the optimal FA detector stated in Thm. \ref{thm:FA_detector}, for the binary case, when the spacing between the two possible release times is $\Delta^{\ast}$. 
The optimal detector, based on $y_{\text{FA}}$ is presented in the following theorem.
\begin{theorem} \label{thm:nonBin_FA_Det}
	The decision rule that minimizes the {\em symbol} probability of error, based on $y_{\text{FA}}$, is given by:
	\begin{align}
	\hat{n}_{\text{FA}}(y_{\text{FA}}) \mspace{-3mu} = \mspace{-3mu} \begin{cases} 
																					0, & n_{\text{floor}}(y_{\text{FA}}) \mspace{-3mu} < \mspace{-3mu} 0 \\
																					2^L \mspace{-3mu} - \mspace{-3mu} 1, & n_{\text{floor}}(y_{\text{FA}}) \mspace{-3mu} \ge \mspace{-3mu} 2^L \mspace{-3mu} - \mspace{-3mu} 1 \\
																					n_{\text{floor}}(y_{\text{FA}}) \\ \mspace{6mu} + \mspace{2mu} \hat{S}_{\text{FA}}^{\tilde{\Delta}}(y_{\text{FA}} \mspace{-4mu} - \mspace{-4mu} n_{\text{floor}}(y_{\text{FA}}) \tilde{\Delta}), & \text{otherwise}.
																				\end{cases} \label{eq:nonBin_FA_Det}
	\end{align}
	
\noindent The probability of error of this detector is exactly $\frac{2^L - 1}{2^{L-1}}$ times the probability of error of the detector in \eqref{eq:errProbSymbBySymbFA} with $\Delta$ replaced by $\tilde{\Delta}$, for $\theta_M$ (in \eqref{eq:errProbSymbBySymbFA}) the solution of \eqref{eq:thetaMEquation}, again, with $\Delta$ replaced by $\tilde{\Delta}$.
\end{theorem}

\begin{proof}
	We first discuss the optimality of the detector in \eqref{eq:nonBin_FA_Det}.
	The extreme cases in \eqref{eq:nonBin_FA_Det} are straightforward, thus, we discuss the ``middle'' points.
	The optimal detector based on $y_{\text{FA}}$ is the ML detector: 
	\begin{align}
		\hat{n}_{\text{FA}}(y_{\text{FA}})  = \argmax_{n \in \{0,1\dots,2^{L}-1\}}f_{Y_{\text{FA}}|X}(y_{\text{FA}}|x = n \tilde{\Delta}). \label{eq:MLfaNonBin}
	\end{align}
	
	\noindent 
	Recall that the density of the FA is unimodal. It follows that $n_{\text{floor}}(y_{\text{FA}}) \tilde{\Delta} + \omega_M \le y_{\text{FA}} \le (n_{\text{floor}}(y_{\text{FA}}) + 1) \tilde{\Delta}$. Since the density is unimodal, we have:
	\begin{align*}
		& f_{Y_{\text{FA}}|X}(y_{\text{FA}}|x \mspace{-3mu} = \mspace{-3mu} n^{\ast} \tilde{\Delta}) \nonumber \\
		& \quad \le \mspace{-3mu} f_{Y_{\text{FA}}|X}(y_{\text{FA}}|x \mspace{-3mu} = \mspace{-3mu} n_{\text{floor}}(y_{\text{FA}}) \tilde{\Delta}), \mspace{7mu} \forall n^{\ast} \mspace{-3mu} \le \mspace{-3mu} n_{\text{floor}}(y_{\text{FA}}), \\
		& f_{Y_{\text{FA}}|X}(y_{\text{FA}}|x \mspace{-3mu} = \mspace{-3mu} n^{\ast} \tilde{\Delta}) \nonumber \\
		& \quad \le \mspace{-3mu} f_{Y_{\text{FA}}|X}(y_{\text{FA}}|x \mspace{-3mu} = \mspace{-3mu} (n_{\text{floor}}(y_{\text{FA}}) \mspace{-3mu} + \mspace{-3mu} 1) \tilde{\Delta}), \mspace{7mu} \forall n^{\ast} \mspace{-3mu} \ge \mspace{-3mu} n_{\text{floor}}(y_{\text{FA}}) \mspace{-3mu} + \mspace{-3mu} 1. 
	\end{align*}
	
	\noindent Thus, in the maximization \eqref{eq:MLfaNonBin} one needs to consider only $n_{\text{floor}}(y_{\text{FA}})$ and $n_{\text{floor}}(y_{\text{FA}}) + 1$. The problem is reduced to a binary detection setup with spacing of $\tilde{\Delta}$. The optimal detector for this problem is given in Thm. \ref{thm:FA_detector}.
	
	For the probability of error we note that the first (in time) and last constellation points exactly correspond to the binary case discussed in Thm. \ref{thm:FA_detector}. The other $2^{L-1}$ constellation points have two adjacent neighbors (a preceding constellation point and a succeeding one). Letting $\theta_M(\tilde{\Delta})$ denote the decision threshold, the probability of error for the ``middle'' constellation points is given by $1 - (\Psi(c,M,\theta_M(\tilde{\Delta}) - \Psi(c,M,\theta_M(\tilde{\Delta}) - \tilde{\Delta})$, where $\Psi(\cdot)$ is defined in \eqref{eq:yFAcdf}. Finally we note that there are $2^{L-1}$ ``middle'' constellation points, thus the over all probability of error is $\frac{2^L - 1}{2^{L-1}}$ times the one stated in \eqref{eq:errProbSymbBySymbFA} (with the proper $\tilde{\Delta}$ and $\theta_M(\tilde{\Delta})$). This concludes the proof.	
\end{proof}

Next, we consider the scaling order of $L$, as a function of $M$, which ensures a vanishing symbol error probability. We note that a linear increase in $L$ results in an exponential decrease in $\tilde{\Delta}$, thus, $L$ should scale at most logarithmically with $M$. 
The next theorem argues that with logarithmic scaling the probability of error does not vanish, thus, a slower scaling is required.
\begin{theorem} \label{thm:nonBinScaling}
	The symbol probability of error of the detector in \eqref{eq:nonBin_FA_Det} vanishes when $L,M \to \infty$, if $L$ scales at most as $\log \log M^{1 -\epsilon}$, for some $\epsilon > 0$. 
\end{theorem}

\begin{proof}
	The proof is provided in Appendix \ref{annex:thm_nonBinScaling_proof}.
\end{proof}

\noindent Thus, to reliably send a large number of bits using the above transmission scheme, a very large $M$ is required. 

Next, we present our numerical results.

\section{Numerical Results} \label{sec:numRes}
	
		\renewcommand{\figWidth}{0.9}
	\begin{figure}[t!]
    \centering
    \includegraphics[width=\figWidth\columnwidth]{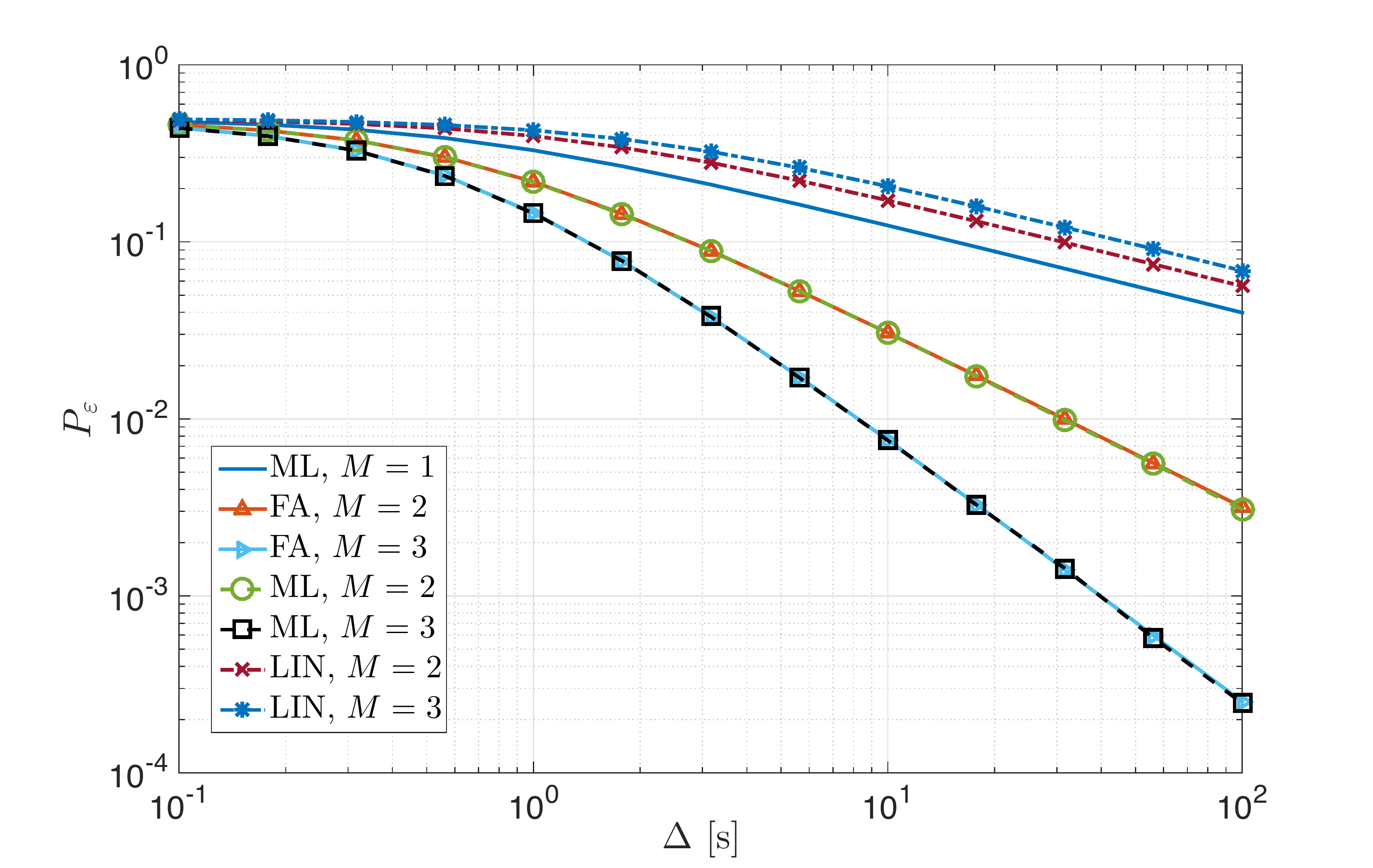}
		\captionsetup{font=footnotesize}
    \caption{$P_{\varepsilon}$ vs. $\Delta$, for $c=1 [s]$ and $M=1,2,3$.}
    \label{fig:PeVsDelta_c1DiffMs}
		\vspace{-0.3cm}
\end{figure}
	
	In this section we numerically evaluate the performance of the different detectors as a function of the channel and modulation.  		
		Fig. \ref{fig:PeVsDelta_c1DiffMs} depicts the probability of error versus different values of $\Delta$, for $M=1,2,3$, for the ML, FA, and linear detectors. Throughout this section $10^6$ trials were carried out for each $\Delta$ point.  
			When $M=1$ all the detectors have identical performance. 
			For larger values of $M$, the probability of error of the ML detector was evaluated numerically, while the probability of error of the FA detector was calculated using \eqref{eq:errProbSymbBySymbFA}. For the linear detector we assumed $w_m = \frac{1}{M}$ which leads to $c_{\text{LIN}} = Mc$. It can be observed that, as expected, the probability of error decreases with $\Delta$. For the ML and FA detectors, the error probability also decreases with $M$, but for the linear detector, the error probability increases with M. Moreover, as stated in Section \ref{subsec:PerfCompare}, Fig. \ref{fig:PeVsDelta_c1DiffMs} shows that the performance of the ML and FA detectors is practically indistinguishable for small values of $M$.  
		
Fig. \ref{fig:PeVsM_c2DiffDeltas} depicts the probability of error versus the number of released particles $M$, for the ML and FA detectors, for $\Delta = 0.2, 0.5$, and $c=2$. Here, $10^6$ trials were carried out for each $M$ point. It can be observed that for small values of $M$, as indicated by Fig. \ref{fig:PeVsDelta_c1DiffMs}, the FA and ML are indistinguishable. 
On the other hand, when $M$ increases, e.g., $M > 100$, the superiority of the ML detector is revealed. This supports the results of Section \ref{subsec:LargeM}.
\begin{figure}[t]
    \centering
    \includegraphics[width=\figWidth\columnwidth]{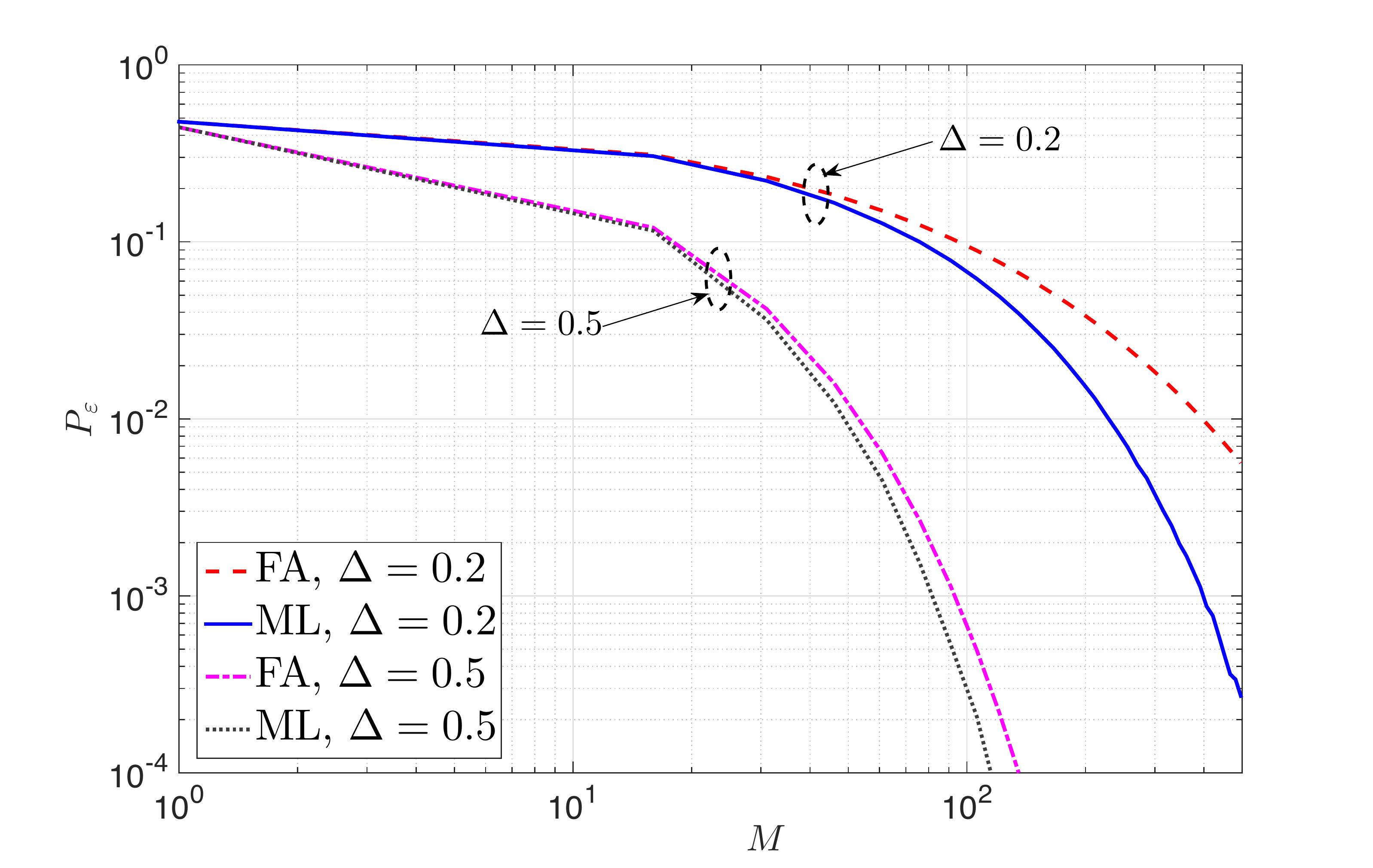}
		\captionsetup{font=footnotesize}
    \caption{$P_{\varepsilon}$ vs. $M$, for $c=2 [s]$ and $\Delta=0.2,0.5 [s]$.}
    \label{fig:PeVsM_c2DiffDeltas}
    \vspace{-0.3cm}
\end{figure}	

Note that Fig. \ref{fig:PeVsM_c2DiffDeltas} also indicates that for large enough $M$, the probability of error decays exponentially with $M$. This implies that if $c$ changes, e.g., the distance between the transmitter and receiver increases, one can achieve the same $P_{\varepsilon}$ by increasing $M$. This is demonstrated in Fig. \ref{fig:MVsC_Pe0p01}, where $P_{\varepsilon}$ is fixed to $0.01$, and the required $M$ is presented as a function of $c$, for different values of $\Delta$. \changeYony{Note that for an uncoded $P_{\varepsilon}$ of $0.01$, coding can be used to drive down the BER to a desired level.}

As discussed in Section \ref{sec:beyondBinary}, one can tradeoff the probability of error with the data rate, i.e. the number of bits conveyed in each transmitted symbol, $L$. 
More precisely, for a given $\Delta$ and $L$, by using $M$ large enough, one can achieve the desired probability of error. This is demonstrated in Fig. \ref{fig:PeVsDelta_c1DiffLs}, which shows that a {\em symbol} probability of error $P_{\varepsilon,s} = 0.01$ can be achieved when $\Delta$ is about 3 seconds by using different $(M,L)$ pairs. This implies that by using a large number of particles, the transmitter can send short messages using a single-shot transmission with relatively small values of $\Delta$. It can further be observed that the required $M$ must scale significantly faster than exponentially with $L$. Yet, to clearly observe the double exponential scaling of $M$ with $L$, much larger values of $(M,L)$ should be considered. \changeYony{Unfortunately, this leads to numerical instabilities in the numerical computations.}
\begin{figure}[t]
    \centering
    \includegraphics[width=\figWidth\columnwidth]{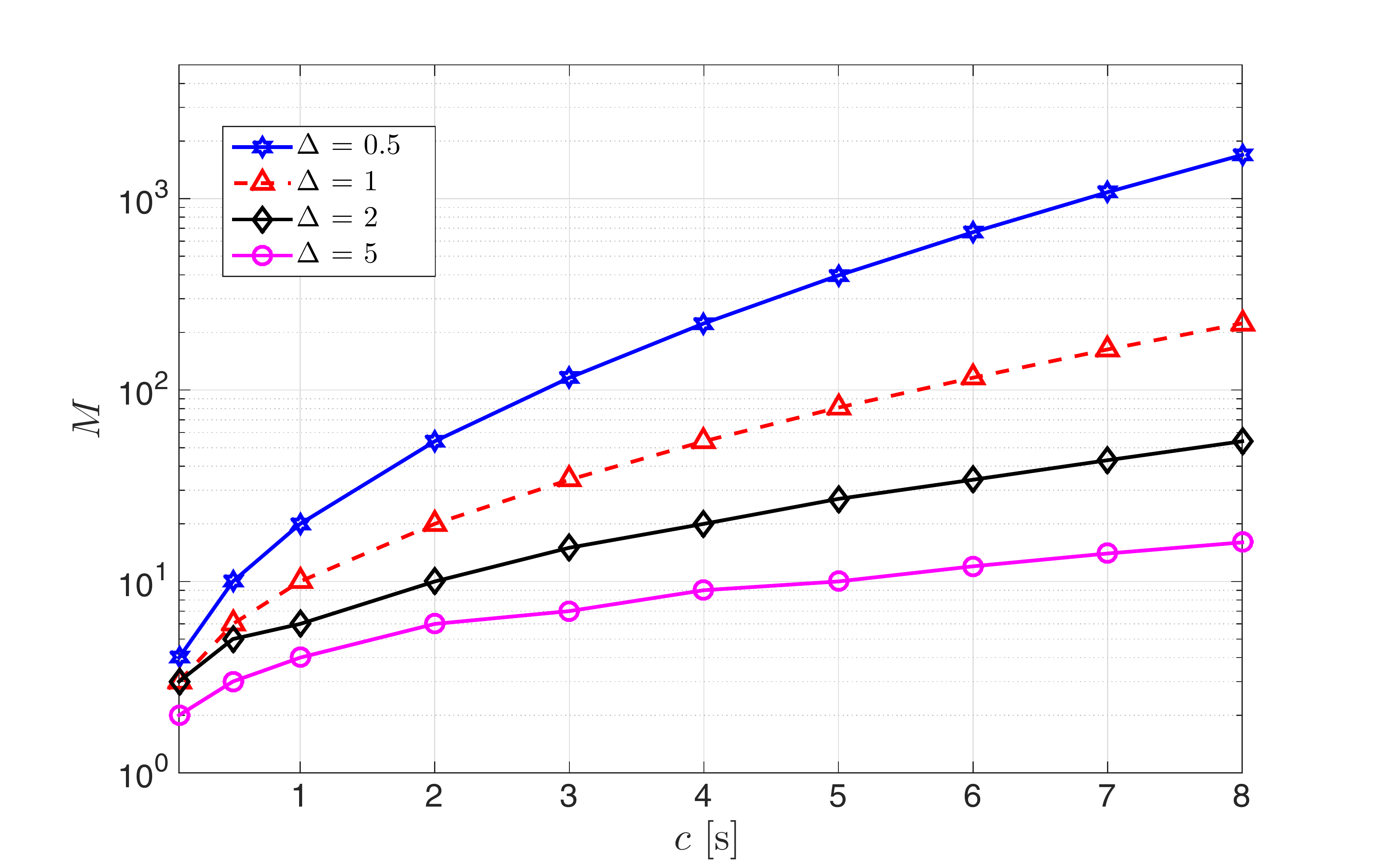}
		\captionsetup{font=footnotesize}
    \caption{The number of particles $M$ required to achieve $P_{\varepsilon}=0.01$, as a function of $c [s]$, for the FA detector.}
    \label{fig:MVsC_Pe0p01}
    \vspace{-0.3cm}
\end{figure}
\begin{figure}[t]
    \centering
    \includegraphics[width=\figWidth\columnwidth]{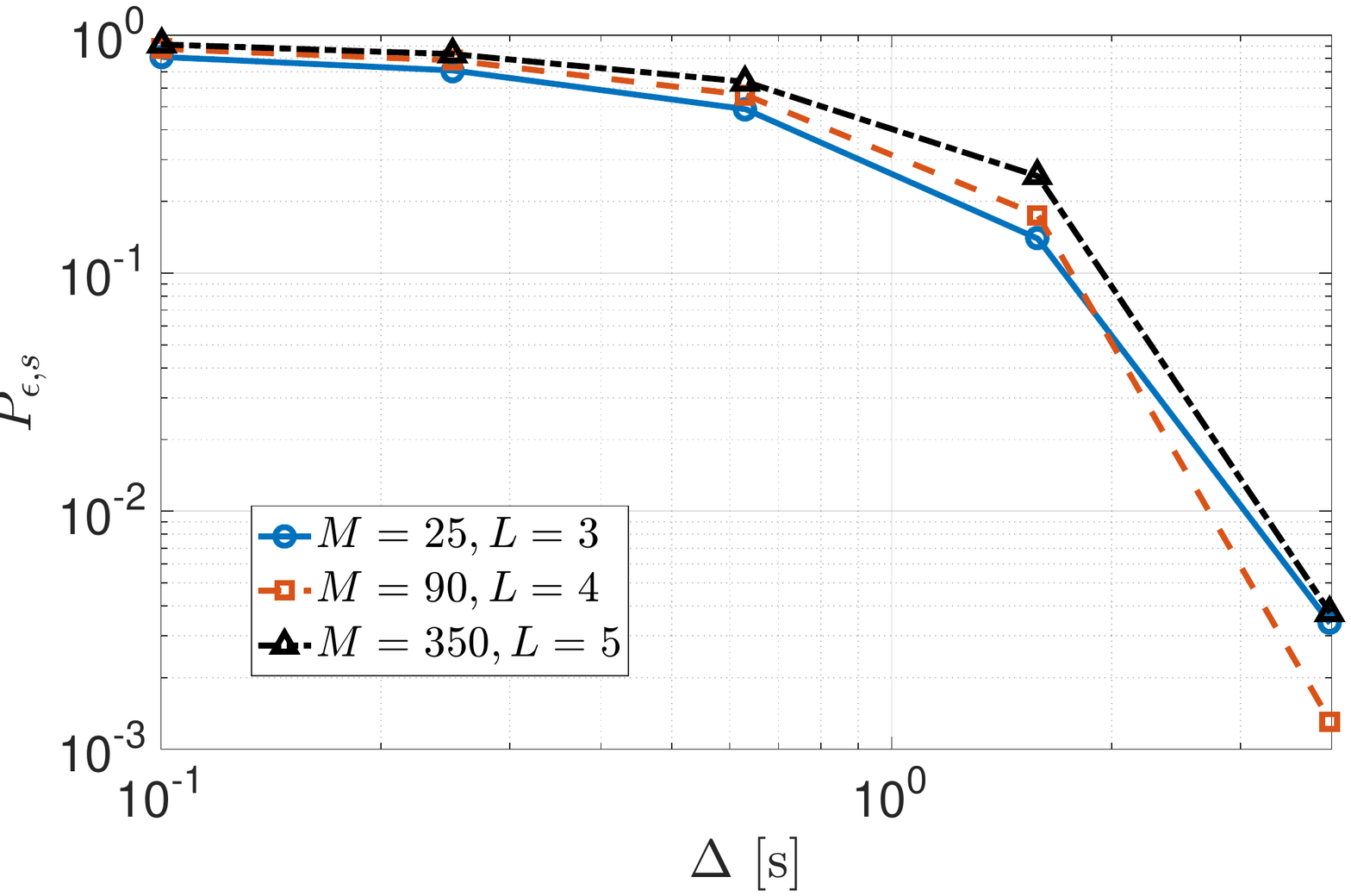}
		\captionsetup{font=footnotesize}
    \caption{$P_{\varepsilon,s}$ vs. $\Delta$, for $c \mspace{-3mu} = \mspace{-3mu} 1 [s]$ and the $(M,L)$ pairs: $(25,3)$, $(90,4)$, $(350,5)$.}
    \label{fig:PeVsDelta_c1DiffLs}
    \vspace{-0.2cm}
\end{figure}

\tyony{Finally, we consider the case of diffusion with a drift, i.e, the AIGN channel. The ML detector for the case of $M>1$ was presented in \cite[eq. (45)]{sri12}, while the FA and linear detectors are discussed in Section \ref{subsec:FA}. Fig. \ref{fig:IG_PeVsVelocity_DifMs} depicts the probability of error versus different values of drift velocity $v$, for the AIGN channel, and for $M=1,2,4$. Here $\Delta = 1[s], D=0.5 [\mu m^2 /s]$ and $d = 1 [\mu m]$, which implies that $\lambda = 1 [s]$. This setting is equivalent to the one simulated in \cite[Fig. 9]{sri12}. $10^6$ trials were carried out for each $v$ point. It can be observed that, similar to the case of diffusion without a drift, the FA detector significantly outperforms the linear detector, and is almost indistinguishable from the ML detector. Only when $M$ is very large the performance gap between FA and ML become apparent. 
\begin{figure}[t]
    \centering
    \includegraphics[width=\figWidth\columnwidth]{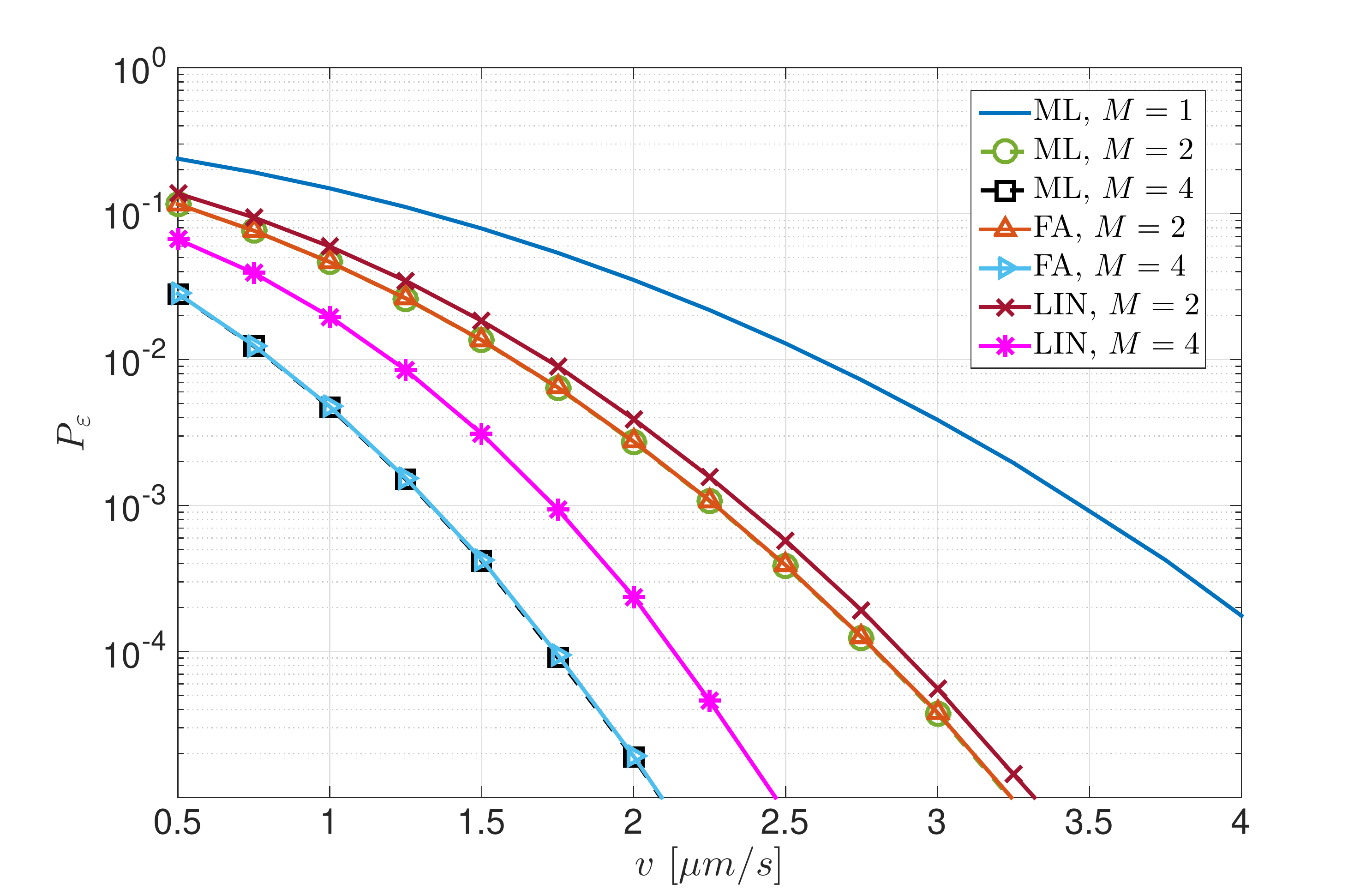}
		\captionsetup{font=small}
    \caption{AIGN channel: $P_{\varepsilon}$ vs. $v$, for $D=2 [\mu m^2 /s]$ and $d = 1 [\mu m]$.}
    \label{fig:IG_PeVsVelocity_DifMs}
    \vspace{-0.3cm}
\end{figure}	
}


\section{Conclusions} \label{sec:conc}

We have studied communication over DBMT channels assuming that multiple information particles are simultaneously released at each transmission. 
We first derived the optimal ML detector, which has high complexity. We next considered the linear detection framework, which was shown to be effective in Gaussian channels or in MT channels in the presence of drift. 
However, we showed that when the noise is stable with characteristic exponent smaller then unity, as is the case in MT channels without drift, then linear processing increases the noise dispersion, which results in a higher probability of error than for the single particle release modulation. 

We then proposed the FA detector and showed that for low to medium values of $M$ it achieves a probability of error very close to that of the ML detector with a complexity similar to that of the linear detector. 
On the other hand, since the first arrival is not a sufficient statistic for the detection problem, it is not expected that the FA and the ML detectors will be equivalent for all values of $M$. To rigorously prove this statement we derived the error exponent of both detectors and showed that, indeed, for large values of M the performance of ML is superior.
\tyony{While the focus of this paper is on DBMT channels without drift, we showed that the above results also extend to the case of diffusion with drift (the AIGN channel). More precisely, for small values of $M$ (up to the order of tens), the FA detector outperforms the linear detector, and closely approaches the performance of the ML detector. Thus, the FA detector provides a very good approximation for the ML detector in DBMT channels.}

\changeYony{Our derivations indicate that the FA detector has a nice property that the conditional densities concentrate towards the particles' release time}, see e.g., Figs. \ref{fig:CondDistributions} and \ref{fig:CondDistributionsIG}. This implies that by using $M$ large enough, one can use large constellations, thus, conveying several bits in each transmission. This property is very attractive for molecular nano-scale sensors that are required to send a limited number of bits and then remain quiet for a long period of time.

\appendices
\numberwithin{equation}{section}


\section{Proof for the Uniqueness of $\theta$ in \eqref{eq:thetaEquation}} \label{annex:Uniqueness_proof}

First we note that the mode of a standard L\'evy-distributed RV is $\frac{c}{3}$, and therefore, the decision threshold must lie in the interval $[\Delta, \Delta + \frac{c}{3}]$. The uniqueness of the solution stems from the fact that the PDF of the L\'evy distribution is unimodal \cite[Ch. 2.7]{zolotarev-book}, and from the fact that the PDFs in the two hypotheses are shifted version of the L\'evy distribution.   
More precisely, note that for $y_1 \to \Delta$, the LHS of \eqref{eq:thetaEquation} tends to zero, while for $y_1 \to \Delta + \frac{c}{3}$ the LHS of \eqref{eq:thetaEquation} is larger than $\frac{c \Delta}{3}$. We now show that the derivative of the LHS of \eqref{eq:thetaEquation}, which is given by $(2y_1 - \Delta) \log \left( \frac{y_1}{y_1 - \Delta} \right) - \Delta$, is positive. This implies that \eqref{eq:thetaEquation} has a unique solution. We write:
\begin{align*}
	& (2y_1 - \Delta) \log \left( \frac{y_1}{y_1 - \Delta} \right) - \Delta \nonumber \\
	& \mspace{80mu} \stackrel{(a)}{=} \Delta \left( \log \left( w \right) \left(1 - \frac{2}{1 - w} \right) - 1 \right) \\
	& \mspace{80mu} \stackrel{(b)}{\ge} \Delta \left( \left(1 - \frac{1}{w} \right) \left(1 - \frac{2}{1 - w} \right) - 1 \right) \\
	& \mspace{80mu} = \Delta \left( \frac{W+1}{W} - 1 \right) \\
	& \mspace{80mu} \ge 0.
\end{align*}

\noindent Here, (a) follows by setting $w = 1 - \frac{\Delta}{y_1}$. Note that since $y_1 \ge \Delta$, then $w \in [0,1]$. For step (b) we use the inequality $1 - \frac{1}{w} \le \log (w)$. Thus, as the derivative is positive, we conclude that \eqref{eq:thetaEquation} has a unique solution in the desired range.

\section{Proofs for Theorem \thmref{thm:PmmUB}} \label{annex:thm_PmmUB_proof}

First we prove that the equation $g(x) = 0$ has a unique solution. Then we derive the properties of $g(x)$, and finally, we derive the upper bound on the mismatch probability.

Let $\alpha = \frac{c \Delta}{3}$. Thus, we can write $g(x)$ as $g(x) = \log \left( \frac{x - \Delta}{x} \right) + \frac{\alpha}{x (x-\Delta)}, x > \Delta$. First, we show that $g(x)$ has a single extreme point which is larger than $\Delta$. Writing the derivative of $g(x)$ we have:
\begin{align*}
	\frac{\partial g(x)}{\partial x} = \frac{\alpha (\Delta - 2x) + \Delta x (x - \Delta)}{x^2(x-\Delta)^2}.
\end{align*}

\noindent Thus, the extreme points of $g(x)$ are the roots of the polynomial $x^2 - \frac{2 \alpha + \Delta^2}{\Delta} x + \alpha$. Plugging $\alpha = \frac{c \Delta}{3}$ and using the expressions for roots of a quadratic equation we obtain that the extreme points are given by:
\begin{align*}
	x_1 &  = \frac{c}{3} + \frac{\Delta}{2} \left( 1 + \sqrt{1 + \frac{4c^2}{9\Delta^2}} \right), \nonumber \\
	x_2 & = \frac{c}{3} + \frac{\Delta}{2} \left( 1 - \sqrt{1 + \frac{4c^2}{9\Delta^2}} \right).
\end{align*} 

\noindent Now, it can be observed that $x_1 > \frac{c}{3} + \Delta > \Delta $ which proves the existence of an extreme point larger than $\Delta$. For $x_2$ we write:
\begin{align*}
	x_2 - \Delta & = \frac{c}{3} + \frac{\Delta}{2} \left( 1 - \sqrt{1 + \frac{4c^2}{9\Delta^2}} \right) - \Delta \nonumber \\
	& = \frac{c}{3} - \frac{\Delta}{2} \left( 1 + \sqrt{1 + \frac{4c^2}{9\Delta^2}} \right) \nonumber \\
	& < 0.
\end{align*}

\noindent Hence, in the range $x>\Delta$, the function $g(x)$ has a single extreme point. Next, we note that $\lim_{x \to \Delta} g(x) = \infty$, while $\lim_{x \to \infty} g(x) = 0$. Therefore, $x_1$ is a minimum point. Thus, the equation $g(x)=0$ has a single solution in the range $x>\Delta$. 


Next, we upper bound the mismatch probability. Let ``mismatch" denote the event of $\hat{S}_{\text{ML}}(\yvec) \neq \hat{S}_{\text{FA}}(\yvec)$. We write:
\begin{align}
	\Pr \{ \text{mismatch} \} & = 0.5 \Big( \Pr \{ \text{mismatch} | x = 0 \} \nonumber \\
	& \mspace{60mu} + \Pr \{ \text{mismatch} | x = \Delta \} \Big).
	\label{eq:Prob_mismatch}
\end{align}

\subsection{Upper Bounding $\Pr \{ \text{mismatch} | x = 0 \}$}
We begin with upper bounding $\Pr \{ \text{mismatch} | x = 0 \}$. Note that if $y_{\text{FA}} \le \Delta$ then $\hat{S}_{\text{ML}}(\yvec) = \hat{S}_{\text{FA}}(\yvec) = 0$, and therefore we analyze only the case of $\Delta < y_{\text{FA}}$.
Recall that for $\Delta < y_{\text{FA}} < \theta_M$ the FA detector decides $\hat{S}_{\text{FA}}(\yvec) = 0$. Hence, a mismatch event occurs when the ML detector declares $\hat{S}_{\text{ML}}(\yvec) = 1$, which occurs if (see \eqref{eq:decisionRuleMultiple}):  
\begin{align}
		\sum_{m=1}^{M}{\log \left( \frac{y_m-\Delta}{y_m} \right) + \frac{c \Delta}{3} \frac{1}{y_m(y_m - \Delta)}} < 0.
		\label{eq:ML_1_decision}
\end{align}

\noindent The LHS of \eqref{eq:ML_1_decision} can be written as:
\begin{align*}
		& \sum_{m=1}^{M}{\log \left( \frac{y_m-\Delta}{y_m} \right) + \frac{c \Delta}{3} \frac{1}{y_m(y_m - \Delta)}} \nonumber \\
		& \qquad= \log \left( \frac{y_\text{FA}-\Delta}{y_\text{FA}} \right) + \frac{c \Delta}{3} \frac{1}{y_\text{FA}(y_\text{FA} - \Delta)} \nonumber \\
		& \qquad \qquad + \sum_{m=2}^{M}{\log \left( \frac{y_m-\Delta}{y_m} \right) + \frac{c \Delta}{3} \frac{1}{y_m(y_m - \Delta)}} .
\end{align*}

\noindent Therefore, \eqref{eq:ML_1_decision} can be written as:
\begin{align}
	& \log \left( \frac{y_\text{FA}-\Delta}{y_\text{FA}} \right)  + \frac{c \Delta}{3} \frac{1}{y_\text{FA}(y_\text{FA} - \Delta)} \nonumber \\
	& \qquad < \sum_{m=2}^{M}{\log \left( \frac{y_m}{y_m-\Delta} \right) - \frac{c \Delta}{3} \frac{1}{y_m(y_m - \Delta)}}.
\end{align}

\noindent Let $\phi(y) = \log \left( \frac{y-\Delta}{y} \right)  + \frac{c \Delta}{3} \frac{1}{y(y - \Delta)}$.
Next, we define the set:
\begin{align*}
	\mathcal{B}_1(y) & \triangleq \Bigg\{(y_2,y_3,\dots,y_M): \nonumber \\
	& \qquad \phi(y) \mspace{-3mu} < \mspace{-3mu} \sum_{m=2}^{M}{\log \left( \frac{y_m}{y_m \mspace{-3mu} - \mspace{-3mu} \Delta} \right) \mspace{-3mu} - \mspace{-3mu} \frac{c \Delta}{3} \frac{1}{y_m(y_m \mspace{-3mu} - \mspace{-3mu} \Delta)}}\Bigg\}.
\end{align*}

\noindent Thus, $\Pr \{ \text{mismatch} | x \mspace{-3mu} = \mspace{-3mu} 0 \}$, when $\Delta \mspace{-3mu} < \mspace{-3mu} y_{\text{FA}} \mspace{-3mu} < \mspace{-3mu} \theta_M$, is given by:
\begin{align}
	& \Pr \{ \text{mismatch} | x = 0, \Delta < y_{\text{FA}} < \theta_M \} \nonumber \\
	& \qquad = \int_{\Delta}^{\theta_M} f_{Y_{\text{FA}|X}}(y|x=0) \nonumber \\
	& \qquad \qquad \times \int_{\mathcal{B}_1(y)}{f_{\{Y_j\}_{j=2}^M|X}(\{y_j\}_{j=2}^M|x=0)} \{dy_j\}_{j=2}^M dy \nonumber \\ 
	& \qquad \stackrel{(a)}{\le} \int_{\Delta}^{\theta_M}{f_{Y_{\text{FA}|X}}(y|x=0) dy} \nonumber \\
	& \qquad = \Psi(c,M,\theta_M) - \Psi(c,M,\Delta), \label{eq:probMM_firstTermBound}
\end{align}
 
\noindent where (a) follows from the fact that in the second integrand is a joint PDF, and therefore it is upper bounded by 1. 
\noindent Following similar arguments, for $\Pr \{ \text{mismatch} | x = \Delta \}$, when $\Delta < y_{\text{FA}} < \theta_M$, we obtain:
\begin{align}
	& \Pr \{ \text{mismatch} | x \mspace{-2mu} = \mspace{-2mu} \Delta, \Delta < y_{\text{FA}} \mspace{-2mu} < \mspace{-2mu} \theta_M \} \nonumber \\
	& \qquad \qquad \le \Psi(c,M,\theta_M-\Delta) - \Psi(c,M,0) . \label{eq:probMM_thirdTermBound}
\end{align}

\subsection{Upper Bounding $\Pr \{ \text{mismatch} | x = 0\}$ for $\theta_M < y_{\text{FA}} $}

First, we recall that when $\theta_M < y_{\text{FA}}$ then $\hat{S}_{\text{FA}}(y_{\text{FA}})=1$. Hence, a mismatch event takes place if the ML detector declares $\hat{S}_{\text{ML}}(\yvec) = 0$, which occurs if (see \eqref{eq:decisionRuleMultiple}):  
\begin{align}
		\sum_{m=1}^{M}{\log \left( \frac{y_m-\Delta}{y_m} \right) + \frac{c \Delta}{3} \frac{1}{y_m(y_m - \Delta)}} > 0.
		\label{eq:ML_0_decision}
\end{align}

\noindent We showed that if $y_{\text{FA}} > x^{\ast}$ then $g(y_{m}) < 0, \forall m=1,2,\dots,M$. In such case the LHS of \eqref{eq:ML_0_decision} is negative and $\hat{S}_{\text{ML}}(\yvec) = 1$, thus, there is no mismatch. Therefore, $\Pr \{ \text{mismatch} | x = 0, \theta_M < y_{\text{FA}} \} = \Pr \{ \text{mismatch} | x = 0, \theta_M < y_{\text{FA}} < x^{\ast} \}$. Now, we define the set:
\begin{align*}
	\mathcal{B}_2(y) & \triangleq \Bigg\{(y_2,y_3,\dots,y_M): \nonumber \\
	& \qquad \phi(y) \mspace{-3mu} > \mspace{-3mu} \sum_{m=2}^{M}{\log \left( \frac{y_m}{y_m \mspace{-3mu} - \mspace{-3mu} \Delta} \right) \mspace{-3mu} - \mspace{-3mu} \frac{c \Delta}{3} \frac{1}{y_m(y_m \mspace{-3mu} - \mspace{-3mu} \Delta)}}\Bigg\},
\end{align*}

\noindent and write $\Pr \{ \text{mismatch} | x = 0, \theta_M < y_{\text{FA}} < x^{\ast} \}$ as:
\begin{align}
	& \Pr \{ \text{mismatch} | x = 0, \theta_M < y_{\text{FA}} < x^{\ast} \} \nonumber \\
	& \qquad = \int_{\theta_M}^{x^{\ast}}f_{Y_{\text{FA}|X}}(y|x=0) \nonumber \\
	& \qquad \qquad \times \int_{\mathcal{B}_2(y)}{f_{\{Y_j\}_{j=2}^M|X}(\{y_j\}_{j=2}^M|x=0)} \{dy_j\}_{j=2}^M dy \nonumber \\ 
	& \qquad \le \int_{\theta_M}^{x^{\ast}}{f_{Y_{\text{FA}|X}}(y|x=0) dy} \nonumber \\
	& \qquad = \Psi(c,M,x^{\ast}) - \Psi(c,M,\theta_M). \label{eq:probMM_secondTermBound}
\end{align}

%
\noindent Following similar arguments, $\Pr \{ \text{mismatch} | x \mspace{-3mu} = \mspace{-3mu} \Delta, \theta_M \mspace{-3mu} < \mspace{-3mu} y_{\text{FA}} \}$ is upper bounded by:
\begin{align}
	& \Pr \{ \text{mismatch} | x = 0, \theta_M < y_{\text{FA}} \} \nonumber \\
	& \qquad \qquad  \le  \Psi(c,M,x^{\ast} - \Delta) - \Psi(c,M,\theta_M - \Delta). \label{eq:probMM_fourthTermBound}
\end{align}

\noindent Combining \eqref{eq:probMM_firstTermBound}, \eqref{eq:probMM_thirdTermBound}, \eqref{eq:probMM_secondTermBound},  and \eqref{eq:probMM_fourthTermBound} we conclude the proof.

\section{Proof of Theorem \ref{thm:ErrExp_FA}} \label{annex:thm_ErrExp_FA_proof}

Let $a_M \mspace{-3mu} \triangleq \mspace{-3mu} \left( 1 \mspace{-3mu} - \mspace{-3mu} \erfc \left( \sqrt{\frac{c}{2\theta_M}} \right) \right)^M$ and $b_M \mspace{-3mu} \triangleq \mspace{-3mu} 1 \mspace{-3mu} - \mspace{-3mu} \left(1 \mspace{-3mu} - \mspace{-3mu} \erfc \left( \sqrt{\frac{c}{2(\theta_M - \Delta)}} \right) \right)^M$. Then, explicitly writing the probability of error in \eqref{eq:errProbSymbBySymbFA}, the error exponent of the FA detector is given by:
	\begin{align}
			\mathsf{E}_{\text{FA}} & = \lim_{M \to \infty} - \frac{\log P_{\varepsilon, \text{FA}}^{(M)}}{M} \nonumber \\
			& = \lim_{M \to \infty} - \frac{\log \left( 0.5 \left( a_M + b_M \right) \right)}{M} \nonumber \\
			& = \min \left\{ \lim_{M \to \infty} - \frac{\log \left( a_M \right)}{M}, \lim_{M \to \infty} - \frac{\log \left( b_M \right)}{M} \right\}. \label{eq:Efa_General}
	\end{align}
	
	Since $\theta_M \to \Delta$, when $M \to \infty$, we write:
	\begin{align}
		& \lim_{M \to \infty} - \frac{\log \left( a_M \right)}{M} \nonumber \\
		& \qquad = \lim_{M \to \infty} - \frac{\log \left( \left( 1 - \erfc \left( \sqrt{\frac{c}{2\theta_M}} \right) \right)^M \right)}{M} \nonumber \\
		& \qquad = - \log \left( 1 - \erfc \left( \sqrt{\frac{c}{2 \Delta}} \right) \right). \label{eq:ErrExp_1stTerm}
	\end{align}
	
\noindent	Next, we analyze the second term in \eqref{eq:Efa_General}, and note that this term depend on the rate of convergence of $\theta_M$ to $\Delta$. Again, we use the fact that $\theta_M \to \Delta$ as $M \to \infty$, and write $\theta_M = \Delta + \delta_M$, where $\delta_M \to 0$. We then characterize the scaling behavior of $\delta_M$ to zero, for large values of $M$. As $\theta_M$ is the decision threshold, by equating the two PDFs in \eqref{eq:yFApdf}, we have the following equality in terms of $\delta_M$:
\begin{align}
	& \left( \frac{\delta_M}{\Delta + \delta_M} \right)^{\frac{3}{2}} e^{-\frac{c}{2} \left(\frac{\Delta}{\delta_M (\Delta + \delta_M)} \right)} \nonumber \\ 
	& \qquad \quad = \left(\frac{1 - \erfc \left( \sqrt{\frac{c}{2 (\Delta + \delta_M)}}  \right) }{1 - \erfc \left( \sqrt{\frac{c}{2 \delta_M} } \right)} \right)^{M-1}. \label{eq:vanish_deltam_1}
\end{align}
	
\noindent Let $M$ be sufficiently large, and recall the equality $\lim_{x \to \infty} \frac{- \log \left(\erfc(x) \right)}{x^2} = 1$. Furthermore, for $M$ large enough we can write $\Delta + \delta_M \approx \Delta$. Thus, we write \eqref{eq:vanish_deltam_1} as:
\begin{align}
	\left( \frac{\delta_M}{\Delta + \delta_M} \right)^{\frac{3}{2}} e^{-\frac{c}{2} \left(\frac{\Delta}{\delta_M (\Delta + \delta_M)} \right)} & \approx \left( \frac{\delta_M}{\Delta} \right)^{\frac{3}{2}} e^{-\frac{c}{2} \left(\frac{1}{\delta_M} \right)} \nonumber \\
	& \approx \left(\frac{\beta}{1 - e^{\frac{-c}{2 \delta_M}}} \right)^{M-1}, \label{eq:vanish_deltam_2}
\end{align}

\noindent where we let $\beta \triangleq 1 - \erfc \left( \sqrt{\frac{c}{2 \Delta}}  \right) \le 1$, and note that $1 - \erfc \left( \sqrt{\frac{c}{2 (\Delta + \delta_m)}}  \right) \approx \beta$. We now assume that $\delta_M$ scales as $\frac{d_1}{M}$, for some constant $d_1$, and show that for large enough $M$ the LHS and RHS of \eqref{eq:vanish_deltam_2} have the same scaling. This also enables finding the constant $d_1$, and calculating the error exponent of the second term in \eqref{eq:Efa_General}. We write \eqref{eq:vanish_deltam_2} as:
\begin{align}
	\left( \frac{d_1}{\Delta} \right)^{\frac{3}{2}} M^{\frac{-3}{2}} e^{-\frac{c}{2d_1} M } & \mspace{-3mu} = \mspace{-3mu} \frac{\beta^{M-1}}{\left( 1 - e^{\frac{-c}{2d_1}M} \right)^{M-1}} \nonumber \\
	& \mspace{-3mu} \approx \mspace{-3mu} \beta^{M-1} \left( 1 \mspace{-3mu} + \mspace{-3mu} (M \mspace{-3mu} - \mspace{-3mu} 1) e^{\frac{-c}{2d_1}M} \right). \label{eq:vanish_deltam_3}
\end{align}

\noindent Thus, by noting that the two sides of \eqref{eq:vanish_deltam_3} must scale to zero at the same rate, we write:\footnote{Note that as we are interested in the error exponent, we apply analysis which focuses only on the scaling law, and therefore we ignore terms which scale slower, e.g., $M^{\frac{-3}{2}}$.}
\begin{align}
	e^{-\frac{c}{2d_1} M } = e^{M(1+\log(\frac{\beta}{e}))} \Rightarrow d_1 = \frac{-c}{2(1+\log(\frac{\beta}{e}))}.
\end{align}

\noindent Having the scaling law of $\delta_M$, we now write the second term in \eqref{eq:Efa_General} as:
\begin{align}
& \lim_{M \to \infty} - \frac{\log \left( b_M \right)}{M} \nonumber \\
& \qquad = \lim_{M \to \infty} - \frac{\log \left( 1 - \left(1 - \erfc \left( \sqrt{\frac{c}{2(\theta_M - \Delta)}} \right) \right)^M \right)}{M} \nonumber \\
%
%
	%
	& \qquad = \lim_{M \to \infty} \frac{ - \log \left(M e^{\frac{-c}{2 d_1} M} \right) }{M} \nonumber \\
	& \qquad = -1 - \log \left(\frac{\beta}{e} \right) \nonumber \\
	& \qquad  = - \log \left( 1 - \erfc \left( \sqrt{\frac{c}{2 \Delta}} \right) \right). \label{eq:ErrExp_2ndTerm}
\end{align}

\noindent Finally, by plugging \eqref{eq:ErrExp_1stTerm} and \eqref{eq:ErrExp_2ndTerm} into \eqref{eq:Efa_General}, we conclude the proof.

\section{Proof of Theorem \ref{thm:nonBinScaling}} \label{annex:thm_nonBinScaling_proof}

We first recall the probability of error of the decision rule \eqref{eq:nonBin_FA_Det}:
\begin{align}
		P_{\varepsilon, \text{FA}} & = \frac{2^L - 1}{2^L} \Bigg( \left( 1 \mspace{-3mu} - \mspace{-3mu} \erfc \left( \sqrt{\frac{c}{2\theta_M}} \right) \right)^M \nonumber \\
		& \mspace{80mu} + \mspace{-3mu} 1 \mspace{-3mu} - \mspace{-3mu} \left(1 - \erfc \left( \sqrt{\frac{c}{2(\theta_M \mspace{-3mu} - \mspace{-3mu} \tilde{\Delta})}} \right) \right)^M \Bigg),
		\label{eq:errProbNonBinFA}
	\end{align}
	
	\noindent where $\theta_M$ is the solution of the following equation in $y_{\text{FA}}$:
	\begin{align}
		& y_{\text{FA}}(y_{\text{FA}} \mspace{-3mu} - \mspace{-3mu} \tilde{\Delta}) \mspace{-3mu} \cdot \mspace{-3mu}  \log \left( \frac{y_{\text{FA}}}{y_{\text{FA}} \mspace{-3mu} - \mspace{-3mu} \tilde{\Delta}} \right)  \nonumber \\
		& \mspace{10mu} + \mspace{-3mu} y_{\text{FA}}(y_{\text{FA}} \mspace{-3mu} - \mspace{-3mu} \tilde{\Delta}) \mspace{-3mu} \cdot \mspace{-3mu} \log \mspace{-3mu} \left( \mspace{-3mu} \frac{1 \mspace{-3mu} - \mspace{-3mu} \erfc \mspace{-3mu} \left( \sqrt{\frac{c}{2(y_{\text{FA}} - \tilde{\Delta})}} \right)}{1 \mspace{-3mu} - \mspace{-3mu} \erfc \mspace{-3mu} \left( \sqrt{\frac{c}{2y_{\text{FA}}}} \right)} \mspace{-3mu} \right)^{\mspace{-10mu} \frac{2(M \mspace{-3mu} - \mspace{-3mu}1)}{3}} \mspace{-12mu} = \mspace{-3mu} \frac{c \tilde{\Delta}}{3}, 
		\label{eq:thetaMEquation_nonBin}
	\end{align}
	
	\noindent As in Appendix \ref{annex:thm_ErrExp_FA_proof}, we let $\theta_M = \Delta + \delta_M$, where $\delta_M \to 0$. 
	Plugging the expression for $\theta_M$ into \eqref{eq:errProbNonBinFA} we write:
	\begin{align}
		P_{\varepsilon, \text{FA}} & \propto \left( 1 - \erfc \left( \sqrt{\frac{c}{2\theta_M}} \right) \right)^M \nonumber \\
		& \qquad + 1 - \left(1 - \erfc \left( \sqrt{\frac{c}{2(\theta_M - \tilde{\Delta})}} \right) \right)^M \nonumber \\
		& \stackrel{(a)}{=} \left( 1 - \erfc \left( \sqrt{\frac{c}{2(\tilde{\Delta} + \delta_M)}} \right) \right)^M \nonumber \\
		& \qquad + 1 - \left(1 - \erfc \left( \sqrt{\frac{c}{2\delta_M }} \right) \right)^M \nonumber \\
		& \stackrel{(b)}{\approx} \left( 1 - \frac{1}{6} e^{-\frac{c}{2} \frac{1}{\tilde{\Delta} + \delta_M}} \right)^M + 1 - \left(1 - \frac{1}{6} e^{-\frac{c}{2} \frac{1}{\delta_M}} \right)^M \nonumber \\
		& \stackrel{(c)}{\approx}  e^{-\frac{M}{6} e^{-\frac{c}{2} \frac{1}{\tilde{\Delta} + \delta_M}}} + 1 - e^{-\frac{M}{6} e^{-\frac{c}{2} \frac{1}{\delta_M}}}, \label{eq:eq:errProbNonBinFA_approx_1}
	\end{align}
	
	\noindent where (a) follows by plugging $\theta_M = \Delta + \delta_M$; (b) follows from the approximation $\erfc(x) \approx \frac{1}{6} e^{-x^2}$ \cite[Eq. (14)]{chiani2003}; and (c) follows from the limit definition of the exponential function. 
	
	Next, recall that $\delta_M$ scales as $\frac{d_1}{M}$ (for the details see Appendix \ref{annex:thm_ErrExp_FA_proof}). Plugging this scaling into the second exponential term in \eqref{eq:eq:errProbNonBinFA_approx_1} we obtain:
	\begin{align}
		1 - e^{-\frac{M}{6} e^{-\frac{c}{2} \frac{1}{\delta_M}}} \approx 1 - e^{-\frac{M}{6} e^{-\frac{c}{2} \frac{M}{d_1}}} \to_{M \to \infty} 0.
	\end{align}
	
\noindent For the first term in \eqref{eq:eq:errProbNonBinFA_approx_1} we note that $\tilde{\Delta} \mspace{-3mu} \approx \mspace{-3mu} 2^{-L}\Delta$ and write:
\begin{equation*}
	e^{-\frac{M}{6} e^{-\frac{c}{2} \frac{1}{\tilde{\Delta} + \delta_M}}} \mspace{-3mu} = \mspace{-3mu} e^{-\frac{M}{6} g(M,L)},
\end{equation*}
 
\noindent where $g(M,L) \mspace{-3mu} \triangleq \mspace{-3mu}  e^{-\frac{c}{2} \frac{1}{2^{-L} \Delta + \frac{d_1}{M}}}$. 
Now, for the probability of error to vanish, it is required that $M \cdot g(M,L) \to \infty$. Thus, we require $g(M,L) \propto M^{-(1-\epsilon)}$ for some $\epsilon > 0$. Explicitly writing this relationship we obtain:
\begin{align}
	\log M^{1-\epsilon} \mspace{-3mu} \propto \mspace{-3mu} \frac{1}{2^{-L} \Delta + \frac{d_1}{M}} \Rightarrow 2^{-L} \Delta \mspace{-3mu} + \mspace{-3mu} \frac{d_1}{M} \mspace{-3mu} \propto \mspace{-3mu} \frac{1}{\log M^{1-\epsilon}}.
\end{align}

\noindent Thus, 
\begin{align}
	\Delta 2^{L} \propto \log M^{1-\epsilon} \Rightarrow L \propto \log \log M^{1-\epsilon}.
\end{align}

\noindent In conclusion, for the probability of error to vanish, $L$ should scale as $\log \log M^{1-\epsilon}$.
	


\end{document}